\documentclass[fleqn,leqno]{CSML}

\def\dOi{13(2:2)2017}
\lmcsheading%
{\dOi}
{1--32}
{}
{}
{Aug.~14, 2015}
{Apr.~28, 2017}
{}


\usepackage{amssymb,hyperref}
\usepackage{amsmath,enumitem}
\newcommand{\ssi}{\makebox[1cm]{iff}}

\numberwithin{equation}{section}

\def\dotminus{\mathbin{\ooalign{\hss\raise1ex\hbox{.}\hss\cr
  \mathsurround=0pt$-$}}}
\newcommand{\squ}{\sqsubseteq}

\newcommand{\pol}{$\pmb {\mathrm{P}}$}
\newcommand{\expo}{$\pmb {\mathrm{EXP}}$}
\newcommand{\pspace}{$\pmb {\mathrm{PSPACE}}$}
\newcommand{\toexp}{$\pmb {\mathrm{TOWER}}$-$\pmb{\mathrm{EXP}}$}
\newcommand{\texp}{$\pmb {\mathrm{TO}}$-$\pmb{\mathrm{EXP}}$}
\newcommand{\primrec}{$\pmb {\mathrm{PRIM}}$-$\pmb{\mathrm{REC}}$}
\newcommand{\C}{$\pmb{\mathrm C}$}
\newcommand{\expc}{$\pmb {\mathrm{EXP(C)}}$}

\newcommand{\n}{\mathbb N}
\newcommand{\q}{\mathbb Q}
\newcommand{\reel}{\mathbb R}
\newcommand{\z}{\mathbb Z}

\newcommand{\espa}{\makebox[1cm]{}}
\newcommand{\esp}{\makebox[.5cm]{}}
\let\esp=\relax

\newtheorem{cl}{Claim}

\newtheorem{df}[cl]{Definition}

\newtheorem*{theo}{An algorithm for $M$}

\newcommand{\point}{${\scriptstyle{\bullet}}\ $}

\newcommand{\inv}{\mathit{Inv}}

\newcommand{\curl}{\preccurlyeq}
\newcommand{\cur}{\prec}
\newcommand{\curllex}{\curl_{\text{llex}}}
\newcommand{\curlex}{\cur_{\text{llex}}}
\newcommand{\strin}{\ensuremath{ \subseteq\!\!\!\!\!\!_{\scriptscriptstyle{/}\ }}}

\newcommand{\words}{\{0,1\}^\ast}
\newcommand{\cantor}{\{0,1\}^{\n}}

\newcommand{\+}[1]{{\mathcal{#1}}}

\begin{document}
\title[Subcomputable Schnorr Randomness]{Subcomputable Schnorr Randomness}

\author[C.~Sureson]{Claude Sureson}
\address{IMJ, CNRS \& Universit\'e Paris 7 Denis Diderot, France.} 
\email{sureson@math.univ-paris-diderot.fr}

\begin{abstract}
The notion of Schnorr randomness refers to computable reals or computable functions. We propose a version of Schnorr randomness for subcomputable classes and characterize it in different ways: by Martin-L\"of tests, martingales or measure computable machines.
\end{abstract}

\maketitle


\section{Introduction }

Martin-L\"of randomness~\cite{ml} is the standard notion of randomness
for infinite binary sequences. Its original definition appeals to
measure but there exist different characterizations based on
Kolmogorov-Chaitin theory of information~\cite{ch},~\cite{levin} or on
the theory of martingales~\cite{sch2,sch3}.

A Martin-L\"of test (written ML test) is a  sequence $(G_n)_{n\in\n}$ of uniformly computably enumerable open subsets of $\cantor$ such that, for each $n\in\n$, the measure  \,$\mu(G_n)$\,  is $\leq 2^{-n}$.\\
An infinite binary sequence $\xi\in\cantor$ is Martin-L\"of random if
for every ML test $(G_n)_{n\in\n}$, $\xi\notin\bigcap_{n\in\n}G_n$ \
($\xi$ avoids all ``effectively null set").

Schnorr viewed this notion as too restrictive and proposed to consider only ML tests $(G_n)_{n\in\n}$ such that the sequence $(\mu(G_n))_{n\in\n}$ is uniformly computable.
He also obtained a characterization in terms of martingales and
orders. More recently Downey and Griffiths~\cite{dogri} characterized
Schnorr's notion using Kolmogorov complexity for ``computable measure
machines" (a prefix-free machine is measure computable if the measure
$\Omega_M$ of the open set generated by the domain of $M$ is
computable).

One thus has (by Schnorr, Downey, Griffiths):
\ \ for any $\xi\in\cantor$, 
 
\espa$\begin{array}[t]{lcl}
\xi \text{ is Schnorr random}&\text{iff}& 
\begin{array}[t]{l}\text{for any ML test }(G_n)_{n\in\n}\text{ with }\\
                   (\mu(G_n))_{n\in\n} \text{ uniformly computable,}\\
                   \xi\notin\bigcap_{n\in\n}G_n.
                   \end{array}\\
         &\text{iff}& 
\begin{array}[t]{l}\text{for any computable martingale } d\\
\text{and any computable order } h,\\
d(\xi\restriction i)< h(i) \text{ \ almost everywhere.}
\end{array}\\
 &\text{iff}& 
\begin{array}[t]{l}\text{for any computable measure machine }M,\\ \text{there is }b\in\n \text{ such that for any }i\in\n\\
K_M(\xi\restriction i)>i-b.
\end{array}
\end{array}$\\

Our work originated from the following question: ``can one recast these results in the primitive recursive framework or in an even  weaker one"?\\
Schnorr showed that one can restrict to ML tests $(G_n)_{n\in\n}$ such that for any $n\in\n$, $\mu(G_n)=2^{-n}$, to define Schnorr's randomness. Similarly Downey and Griffiths proved that one can restrict to machines $M$ such that \,$\Omega_M=1$. Hence the natural amendments do not work (requiring $(\mu(G_n))_{n\in\n}$ to be ``uniformly primitive recursive" or \,$\Omega_M$\, to be a ``primitive recursive real"). But for a subcomputable class \,\C\, of functions, one can nevertheless define a notion of ML-\,\C\,-S(chnorr) test and a notion of measure \,\C\, computable machine (one focuses on the pace of obtention of the measure).\\
If in the martingale formulation, we allow all martingales and orders in \,\C,  the notion will be too strong and we shall not obtain equivalence with the other characterizations. Let $h$ be a computable order. Then its inverse $\text{Inv}_h$ (defined as \,$\text{Inv}_h(n)=\text{least }k\ h(k)\geq n$) is also computable. It is not true anymore for primitive recursive functions. Hence,  relatively to a class \,\C\, of functions, we shall call an order $h$  a true \,\C\,-order if both $h$ and $\text{Inv}_h$ belong to \,\C.

In this article, we shall study the relations between the three following notions: for 
$\xi\in\cantor$,
\begin{itemize}
\item $\xi$ is ML-\,\C\,-S random \ssi $\xi$ passes all ML-\,\C\,-S tests.
\item $\xi$ is Kolmogorov-\,\C\,-S random \ssi  for any measure \,\C\ computable machine $M$, there is $b\in\n$ such that for any $n\in\n,\ K_M(\xi\restriction n)> n-b$.
\item  $\xi$ is martingale-\,\C\,-S random \ssi $\left(\begin{array}{l}
\text{for any martingale }d:\words\to\q_2\ \text{ in \,\C,}\\ \text{and any true \,\C\,-\,order }h, \ d(\xi\restriction i)< 2^{h(i)} \ \text{a.e.}
\end{array}\right.$\\
\end{itemize}

 We show that if \,\C\, is the class of primitive recursive functions or the class  \,\pspace, then these three notions coincide.

One can check by using \,ML-\,\C\,-S tests that (ML) \primrec -S randomness is strictly weaker than Schnorr randomness. But the martingale approach is better suited to separate the different notions of randomness as \,\C \ varies among time-complexity classes. We shall thus rely on the important amount of work centered around the martingale tool and the associated notions of subcomputable randomness. This is the field of Resource Bounded Randomness initiated by Lutz and developed by  Ambos-Spies, Lutz, Mayodormo,  Wang and other people.\\
We shall compare our notion of martingale-\,\C -S randomness with Lutz~\cite{lu1,lu2} notion of p-randomness, with Wang's \,(\pol,\pol)-S randomness and with Buss, Cenzer and remmel~\cite{buss} weaker notion of BP-randomness  (~\cite{buss} results about primitive recursiveness have been a strong motivation to us). 
Wang's notion is a version of Schnorr randomness for the class of polynomial time computable functions, the martingales are required to be in \,\pol\  and  all orders in \,\pol\ \,are allowed. (By a delaying computation argument) this is the same as allowing all computable orders. Our concern with the status of the inverse of the order weakens the notion and enables more variety inside the set of computable sequences.
Building on techniques of Wang and results of Schnorr, we show that one can obtain a whole hierarchy.

The two following tableaux summarize the situation: \,\C \ randomness is the analog for the class \,\C\ of computable randomness, martingale-\,\C -S randomness is abbreviated to \,\C -S randomness, and \,\C -W randomness stands fo weak (Kurz) randomness with regard to the class \,\C.
Implications in the tableaux cannot be reversed and in the second tableau, this holds even when restricting to the class of computable infinite sequences.\\

\noindent$\begin{array}{|ccccc|}\hline
&&&&\vspace{-4mm}\\
\!\text{Computable randomness}\!\!\!&\!\!\!\overset{\esp}{\Rightarrow}&\text{Schnorr randomness}&\Rightarrow&\text{weak randomness}
\\
&&&&\vspace{-4mm}\\
\Downarrow\esp &&\Downarrow\esp&&\Downarrow\esp
\\
&&&&\vspace{-4mm}\\
\!\!\!\!\text{\primrec\ randomness}\!\!\!\!\!\!&\!\!\!\!\!\!\!\Rightarrow\!\!\!\!&\!\!\!\!\text{\primrec\,-S randomness}\!\!\!\!&\!\!\!\!\Rightarrow\!\!\!\!&\!\!\!\!\text{\primrec\,-W randomness}\!\!\\
\hline
\end{array}$\\

\centerline{Tableau 1}

\noindent$\begin{array}{|ccccc|}
\hline
&&&&\vspace{-4mm}\\
\text{\primrec\ randomness}\!\!\!&\!\!\!\!\!\Rightarrow\!\!\!\!&\!\!\!\!\text{\primrec\,-S randomness}\!\!\!\!&\!\!\!\!\Rightarrow\!\!\!\!&\!\!\!\!\text{\primrec\,-W randomness}\!\!\\
\Downarrow &&\Downarrow&&\Downarrow\\
\text{\expo\ randomness}&\Rightarrow&\text{\expo\,-S\ randomness}&\Rightarrow&\text{\expo\,-W randomness}\\
\Downarrow &&\Downarrow&&\Downarrow\\
&&&&\vspace{-3mm}\\
\text{\pol\ randomness}&\Rightarrow&\text{\pol\,-S\ randomness}&\Rightarrow&\text{\pol\,-W randomness}\\
\hline
\end{array}$\\

\centerline{Tableau 2}

\nocite{ko,ch,lu2}

\section{A few classical definitions.}
 
\subsection{Some notation.}

$\n,\,\q,\,\q_2,\,\reel$ denote respectively the set of natural, rational, dyadic rational and real numbers (dyadic rational numbers are  of the form \ $m2^{-n}$, for $m\in\z$ and $n\in\n$).

$\words$ is the set of finite binary sequences (or strings on
$\{0,1\}$) and\, $\cantor$ is the set of infinite binary sequences.
\begin{itemize}
\item If $x$ is a finite sequence, then $|x|$ represents its
  length. For an integer $i\in\n,\ x\restriction i$ is the restriction
  of $x$ onto the set $\{0,1,\ldots,i-1\}$.

\noindent We consider the (prefix) partial ordering $\curl$ defined on finite
binary sequences by
\[\mbox{\ $x\curl y$ \ssi $x$ is a prefix of $y$ \ \ (that is \
  iff \ $|x|\leq|y|$ and $y\restriction |x|=x$)\,.}
\]

\noindent We shall also use the well-ordering $\curllex$ (length-lexicographic
ordering):
\[\mbox{$x\curllex y \ssi \begin{cases}
|x|< |y|\text{ \,or}\\
|x|=|y|\text{ and  } x \text{ is before }y \text{ in lexicographic order.}
\end{cases}$}
\]
\noindent $\cur$ (respectively $\curlex$) denotes the corresponding strict ordering.

\item Now for $\alpha\in\cantor$ and $i\in\n$, we also write $\alpha\restriction i$ for the restriction of $\alpha$ onto the set\linebreak $\{0,1,\cdots,i-1\}$. If $x\in\words$, the notation $x\curl\alpha$ means $\alpha\restriction |x|=x$.\\
If $x,y\in\words,\ i\in\{0,1\},\ \alpha\in\cantor$, we write $xy,\ xi,\ x\alpha$ for the corresponding concatenation.

\item Given a finite set $X$, $|X|$ is the number of elements of $X$. To avoid confusion, for $r\in\reel$, we shall  write \,$\Vert r\Vert$\, to mean the absolute value of $r$.\\
The function \ $\langle\ ,\ \rangle:\n\times\n \to\n$ \ is the
classical polynomial time bijection  defined as
\[\langle m,n\rangle=m+(m+n)(m+n+1)/2\mbox{, for }m,n\in\n\,.\]
 Let \,$(\ )_0,\,(\ )_1:\n\to\n$ \,denote the (polynomial time) inverse functions: for $i\in\n$, $\langle(i)_0,(i)_1\rangle\,=\,i$.

\item Our references in Recursion Theory are \cite{odi1,odi2}, and in Algorithmic Randomness, we rely on ~\cite{dohi} and~\cite{nies}. We thus write $K_M$ for Kolmogorov complexity when considering a prefix-free Turing machine $M$ (see ~\cite[Ch.3.5]{dohi}).\\
The terms ``recursive" and ``computable" have similar meanings. 

\item  Concerning  topology and measure, we consider the classical product topology on $\cantor$ (see ~\cite{dohi,nies}). If  $x\in\words$, then we denote by $[x]$ the basic open set $\{x\alpha:\alpha\in\cantor\}$ and if  $X\subseteq\words$, $[X]$ is the open subset of $\cantor$ generated by $X$, that is
\[[X] = \{x\alpha : x\in X,\ \alpha\in\cantor\}\,.\]
 $\mu$ is the uniform measure on $\cantor$: if $x\in\words$, then $\mu([x])=2^{-|x|}$. When computing measure, we shall always deal with open (and hence measurable) sets.

\item Generally, we use lowcase greek letters $\alpha,\xi...$ for infinite binary sequences and lowcase roman letters $x,y...$ for finite sequences.
\end{itemize}

\subsection{Schnorr Randomness.}

 We recall here the definition of Schnorr randomness and give three different characterizations (due to Schnorr, Downey and Griffiths). For the definitions of a ``computable real" or of a ``computable (real valued) function", we refer to ~\cite[5.1 and 5.2.1]{dohi}.

\begin{defi}\hfill
\begin{enumerate}[label=\({\alph*}]
\item A sequence $(G_n)_{n\in\n}$ of open subsets of $\cantor$ is a Martin-L\"of test (abbreviated as ML test) if there is a recursively enumerable set $X\subseteq\n\times\words$ such that setting, for $n\in\n$, \ $X_n=\{x\in\words : (n,x)\in X\},$ one has \ 
$G_n=[X_n]$ and $\mu(G_n)\leq 2^{-n}$.

\item A sequence $\xi\in\cantor$ passes  the ML test $(G_n)_{n\in\n}$ if \ $\xi\notin\bigcap_{n\in\n}G_n$ \ (otherwise it fails the test).

\item A sequence $\xi\in\cantor$ is random if it passes all ML tests.
\end{enumerate}
\end{defi} 

\noindent Schnorr viewed this notion of randomness as too strong and proposed the following:

\begin{defi}[\textnormal{\ Schnorr}]  \label{schno-historic}\hfill
\begin{itemize}[label=$-$]
\item A Schnorr test is an ML test $(G_n)_{n\in\n}$ \ such that $\mu(G_n)$ is uniformly computable in $n$.

\item A sequence $\xi\in\cantor$ is Schnorr random if it passes all
  Schnorr tests.
\end{itemize}
\end{defi}

\noindent There is a characterization of Schnorr randomness in terms of martingales. We recall:

\begin{defi}\hfill
\begin{enumerate}[label=\({\alph*}]
\item A function $d:\words\to\reel^+$ is a martingale if for any $x\in\words, \ d(x0)+d(x1)=2d(x)$.

\item A function $h:\n\to\n$ is an order if it is nondecreasing and unbounded.
\end{enumerate}
\end{defi}

\begin{thm}[\textnormal{~\cite{schno}}]
A sequence $\xi\in\cantor$ is Schnorr random iff for any computable martingale $d$ and any computable order $h$, \ $d(\xi\restriction n)< h(n)$ \ a.e. (a.e. stands for ``almost everywhere").
\end{thm}

The last characterization we shall consider in this article is more recent and due to Downey and Griffiths.

\begin{defi}\hfill
\begin{itemize}[label=$-$]
\item If $M$ is a Turing machine, then \,$\Omega_M$\, is the measure $\mu([\text{dom}(M)])$.

\item A prefix-free machine $M$ is called a computable measure machine
  if \ $\Omega_M$ \ is a computable real.
\end{itemize}
\end{defi}

\begin{thm}[\textnormal{~\cite{dogri}}]\label{dogri}
A sequence $\xi\in\cantor$ is Schnorr random iff for each computable measure machine $M$, there is \,$b\in\n$\, such that for any $n, \ K_M(\xi\restriction n)>n-b$.
\end{thm}

When trying to extend these definitions to subcomputable classes, one must be cautious. For instance, the notion of ``primitive recursive real" is problematic (see ~\cite{chen}). We shall  consider Cauchy style definitions (rather than left cut ones) and all the definitions which have been omitted in this review paragraph, will be provided.

\section{Notions of Resource bounded Schnorr randomness.}

\esp An important body of concepts and results (now classical) has been obtained by  Ambos-Spies, Ko, Lutz, Mayodormo, Wang and others (see~\cite{ambomayo} for a survey). Our original motivation came from primitive recursiveness and the article of Buss, Cenzer and Remmel~\cite{buss}. We shall thus build on all these works to propose here, relatively to a class of functions, three possible characterizations (inspired from the  previous section) of Schnorr resource bounded randomness: in terms of Martin-L\"of tests, Kolmogorov complexity and martingales. (Depending on the chosen class) we shall study when these different approaches lead to the same notion. Later in the paper, we shall also compare these definitions with a different concept proposed by Wang~\cite{wang2000}.

To motivate our definitions, let us note that some properties are given for free when dealing with recursive functions. For instance, if $f:\n\to\n$ is recursive and unbounded, its inverse $\inv_f$ (defined by $\inv_f(n)=\text{least }k\ f(k)\geq n$) is also recursive. It is not true anymore for primitive recursive functions: there exists a primitive recursive function whose inverse is the Ackermann function \
 $n\mapsto A(n,n)$. Hence our definitions will have to incorporate new conditions.
 
 \subsection{Definitions.}

 We shall consider  time-complexity classes \,\C\, of functions of the form: 
\[\hbox{\C\ }=\ \bigcup_{f\in F_C}\mathrm{FDTIME(f(n))}\,,\]
 for \,$F_C$\, a class with appropriate closure properties of time-constructible functions.

Such classes \,\C\, are:
 
\begin{defi}\hfill
\begin{enumerate}
\item \pol $\ =\ \bigcup_{k\in\n}\text{FDTIME}(n^k)$,

\item \expo $\ =\ \bigcup_{k\in\n} \text{FDTIME}(2^{n^k})$,

\item Let $T:\n\times\n\to\n$ be the function recursively defined  by  
\[\begin{array}[t]{l}
  T(0,n)=n \\
  T(k+1,n)=2^{T(k,n)}.
  \end{array} \]
  Then \ \toexp\ $\ =\ \bigcup_{k\in\n}\text{FDTIME}(T(k,n))$.

\item  \primrec\ is the class or recursively primitive functions.
\end{enumerate}
\end{defi}
 
 To view \primrec\ as a time-complexity class may require some justification: an easy modification of ~\cite[Thm VIII 8.8]{odi2}) gives:
 
 \begin{lem}
A (total) function $f$ is primitive recursive iff it can be computed by a Turing machine in time $\mathcal O(g(n))$, for $g$ primitive recursive and time-constructible.
\end{lem}

For example, if \,\C\,=\,\pol, then we take for \,$F_C$\, the class of polynomial functions with coefficients in $\n$; if \,\C\,=\,\primrec, then we consider for \,$F_C$\, the class of primitive recursive functions which are time-constructible.

\begin{defi}
If \ \C\  \ $=\ \bigcup_{f\in F_C}\text{FDTIME}(f(n))$, with \,$F_C$\,\ as above, then let \ \ \expc \ \ $=\ \ \bigcup_{f\in F_C}\text{FDTIME}(2^{f(n)})$.
\end{defi}

For instance, \ \expo \ = \expo(\pol) \ and \ \C\ =\ \expo(\C)\        
 when \ \C\,\ is \ \toexp\ \,or\, \primrec \ \,(all this will allow us to state our results in a unified way).
 We shall also consider the classical space related class:
 
 \begin{defi}
 \pspace\ \ =\ \ $\bigcup_{k\in\n}\text{FSPACE}(n^k)$.
 \end{defi}
 
 When dealing with functions \,$f:\n\to\n$,
  martingales \,$d:\words\to\q _2$ 
  \,or approximations of \,$\reel$
  valued martingales \,$g:\words\times\n\to\q_2$, to decide whether these functions belong to one of the above classes \C, we must fix a representation of the different inputs and outputs, and hence a measurement of their size.
 \begin{itemize}
 \item integers will be under unary representation: $n\in\n$ is thus viewed as $1^n$ and its size is $n$,
 \item strings $x\in\words$ have classically size $|x|$,
 \item there is a constant $\theta\in\n$ such that dyadic rational numbers  of the form $m2^{-n}$, for $m\in\n$ and $n\in\n$, are (reasonably) coded by a string in $\words$ of length \,$\leq \theta(\log(m)+n)$
\end{itemize}

\noindent To introduce subcomputable Martin-L\"of tests, let us state a few definitions. We often identify a Turing machine $M$ with the (partial) recursive function it computes. We write $M(x)\downarrow$ \ to mean that the machine $M$ halts on input $x$ (yielding as output whatever is written on a dedicated tape).

\begin{defi} \label{def-basic}\hfill
\begin{enumerate}[label=\({\alph*}]
\item Let $M$ be a Turing machine. Then for $t\in\n,\ x,y\in\words$, we set
\begin{itemize}[label=$-$]
\item $M_t(x)=y \ssi (\begin{array}{l}
              M(x)\downarrow \text{\,and on input }x,\ M\\
              \text{ outputs }y\text{ in at most }t\text{ steps.}
              \end{array}$

\item $M^{\text{space}}_t(x)=y \ssi (\begin{array}{l}
              M(x)\downarrow \text{\,and on input }x,\ M\\
              \text{ outputs }y\text{ having used at most }t\text{ cells}\ \ (t\geq |x|,|y|).
              \end{array}$
              \end{itemize}

\item Given a recursively enumerable set $X\subseteq\n\times\words$ \,and a machine $M$ such that $X=\text{dom}(M)$, we set for $m,t\in\n$,
\begin{itemize}[label=$-$]
\item
 $X_m=\{x\in\words : (m,x)\in X\}=\{x\in\words : M(m,x)\downarrow\}$
\item $X^M_{m,t}=\{x\in\words : (m,x)\in \text{dom}(M_t)\},$
\item $X^{M,\text{space}}_{m,t}=\{x\in\words : (m,x)\in
\text{dom}(M^{\text{space}}_t)\}.$
\end{itemize}
\end{enumerate}
\end{defi}

\noindent Let $(G_n)_{n\in\n}$ be a Schnorr test (Definition~\ref{schno-historic}): the sequence $(\mu(G_n))_{n\in\n}$ is thus uniformly computable. This implies the existence of a computable function $F:\n\times\n\to\q_2$ such that for any $i,n\in\n,\ \Vert \mu(G_n)-F(n,i)\Vert \ \leq 2^{-i}$.

Hence a natural attempt to extend the notion of ML Schnorr  test to the primitive recursive context would be to require $F$ to be primitive recursive and to consider as random, infinite sequences which pass all such tests . This is doomed:

\begin{thm}[\textnormal{~\cite{schno}}]\hfill
Let $(G_n)_{n\in\n}$ be a Schnorr test. Then there exists a Schnorr test \,$(O_n)_{n\in\n}$ such that for any $n$, $\mu(O_n)=2^{-n}$ and \ $\bigcap_{n\in\n}G_n\,\subseteq\,\bigcap_{n\in\n}O_n$.
\end{thm}

\noindent Hence ML tests $(G_n)_{n\in\n}$ with  $\mu(O_n)=2^{-n}$,  for  $n\in\n$, suffice to define Schnorr randomness.

Another try consists in noting (it is implicit in several classical proofs) that if $(G_n)_{n\in\n}$ is a Schnorr test -      associated with $X\subseteq\n\times\words$ and a machine $M$ - then using the approximating function $F$ above,
 one can check the existence  of a computable function $f:\n\times\n\to\n$ such that 

\espa\point for any $n,i\in\n, \ \mu([X_n])-\mu([X^M_{n,f(n,i)}])\leq 2^{-i}$.

We could thus require $f$ to be primitive recursive and consider the following definition:\\

\begin{defi} ~\label{def-ML-test}\hfill
\begin{itemize}[label=$-$]
\item Let \ \C\ \ be one of our time-complexity classes. An ML test $(G_n)_{n\in\n}$ is called an ML-\,\C\,-S test if there exist a recursively enumerable set $X\subseteq\n\times\words$ associated with a machine $M$ such that \ $X=\text{dom}(M)$ and a function \ $f:\n\to\n$ \ in \ \C \,, called the controlling function, such that for any $m,i\in\n$,
\begin{itemize}[label=$\circ$]
\item $G_m=[X_m]$,
\item $\mu([X_m])-\mu([X^M_{m,f(m+i)})]) \ \leq \ 2^{-i}$ \ 
\end{itemize}
\item For \ \C\ =\ \pspace, we require $f:\n\to\n$ to be in \,\pspace\, and to satisfy for any $m,i\in\n$,
 \ \,$\mu([X_m])-\mu([X^{\text{M,space}}_{m,f(m+i)}]) \ \leq \ 2^{-i}$
 \end{itemize}
\end{defi}

\noindent(The ``S" in ML-\,\C\,-S-test stands for ``Schnorr")

Concerning an extension of the Downey-Griffiths characterization,  we could restrict to prefix-free machines $M$ such that \ $\Omega_M$ \ is a primitive recursive real. But again, this does not produce a new notion since by ~\cite{dogri}, to characterize  Schnorr randomness, one can restrict to machines $N$ with \ $\Omega_N=1$. Hence we proceed as in the previous definition.

\begin{defi}\label{defi-c-measure}
Given a prefix-free machine $M$ and $t\in\n$, we set \\ \centerline{$\Omega_{M_t}=\mu([\text{dom}(M_t)])$ \ and \ $\Omega_{M^{\text{space}}_t}=\mu([\text{dom}(M^{\text{space}}_t)]))$.}
\begin{itemize}[label=$-$]
\item For \ \C\ \ one of our time-complexity classes, a prefix-free machine $M$ is termed ``measure \,\C\ computable" if there is a function $g:\n\to\n$ \,in \ \C, (also) called the controlling function, such that for any $i\in\n$, \ $\Omega_M - \Omega_{M_{g(i)}}\leq 2^{-i}$.
\item If \ \C\ =\ \pspace, we require the existence of $g$ in \,\pspace\,  such that  $\Omega_M - \Omega_{M^{\text{space}}_{g(i)}}\leq 2^{-i}$.
\end{itemize}
\end{defi}

\begin{rem}\label{controlling}\hfill
\begin{itemize}[label=$-$]
\item Given \,\C\, a time complexity-class, since integers are under unary representation, one gets the same notion of test or of measure \C\,-computability by requiring the controlling function to be in \,\C\, or in \,$F_C$.
\item Similarly for \,\pspace\,, one can indifferently require the controlling function to be a polynomial function (with coefficients in $\,\n$), a function in \,\pol\ \,or a function in \,\pspace.
\end{itemize}
\end{rem}

\noindent As mentioned earlier, part(ii) of the following definition is redundant in the recursive case:

\begin{defi}\label{def-true-order}\hfill
\begin{enumerate}[label=\({\roman*}]
\item Given an unbounded function \,$f:\n\to\n$, one defines the ``inverse" of $f$ as follows: 
 \esp for $n\in\n$,
\ $\text{Inv}_f(n)=\,\text{least $k$ s.t. }f(k)\geq n$.
\item Let \,\C\ \,be one of our complexity classes, an order $h$ 
is a true \,\C\,-order if both $h$ and $\text{Inv}_h$ belong to \,\C.
\end{enumerate}
\end{defi}

\noindent In the subcomputable framework, to define martingale related randomness, we shall restrict to $\q _2$-valued martingales (this is not absolutely necessary, one can consider $\reel$-valued martingales which are\,\ \C\,-\,approximable, as does Lutz~\cite{lu1} for \ \C\,=\,\pol\,).

We now state the respective definitions of \,\C\,-Schnorr randomness. \newpage

\begin{defi}\label{def-csrandom}
Let \ \C\ \,be one of our complexity classes and let\, $\xi\in\cantor$.
\begin{enumerate}[label=\({\alph*}]
\item $\xi$ \,is ML-\,\C\,-S random \ssi $\xi$ passes all ML-\,\C\,-S tests.
\item $\xi$ \,is Kolmogorov-\,\C\,-S random \ssi  for any measure \,\C\ computable machine $M$, there is $b\in\n$ such that for any $n\in\n,\ K_M(\xi\restriction n)> n-b$.
\item \ $\xi$ \,is martingale-\,\C\,-S random \ssi $\left(\begin{array}{l}
\text{for any martingale }d:\words\to\q_2\ \text{ in \,\C,}\\ \text{and any true \,\C\,-order }h, \ d(\xi\restriction i)< 2^{h(i)} \ \text{a.e.}
\end{array}\right.$
\end{enumerate}
\end{defi}

\begin{rem}
If \,\C\ \,is \,\toexp\ \,or\, \primrec, \,then the condition ``$d(\xi\restriction i)< 2^{h(i)} \ \text{a.e.}$" gives the same notion of randomness as the usual one ``$d(\xi\restriction i)< h(i) \ \text{a.e.}$".
\end{rem}

\noindent We shall first study the relation between Martin-L\"of and Kolmogorov complexity notions of randomness, and later the link between Martin-L\"of and martingale notions of randomness.

\subsection{The relation between the Martin-L\"of and the Kolmogorov complexity notions.}

In our definition of ML-\,\C\,-\,S tests, we did not require the generating sets to be prefix-free. To obtain this in a uniform way, we shall resort to the classical argument showing that a recursively enumerable generating set can be replaced by a recursive prefix-free one (see~\cite[1.8.26]{nies}). We propose here a quadratic time algorithm (or linear space) algorithm yielding the new generating set.

\begin{clm}\label{claim-pf-quadra}
Let $X\subseteq\n\times\words$ and let $M$ be a machine such that \ $\text{dom}(M)=X$. Then one can define a set \,$Y\subseteq\n\times\words$ and a machine $N$ such that:
\begin{enumerate}[label=\({\alph*}]
\item
 \ $Y\,=\,\{(n,x)\in\n\times\words: N(n,x)\downarrow\}$ \,and there is a constant $d\in\n$ such that \\
\centerline{for any $(n,x),\ \ \,N(n,x)\downarrow \ \  \Leftrightarrow \ N \text{ on input }(n,x) \text{ halts in at most } d(n+|x|)^2 \text{ steps}$.}
\item \ Setting for $n,t\in\n,\ \left(\begin{array}{l}
Y_n\,=\,\big\{x\in\words: N(n,x)\downarrow\big\} \ \text{and}\\
  Y_n(t)\,=\,\big\{x\in\words : |x|=t \text{ and } x\in Y_n\big\},
\end{array}\right.$\\
  one has \ $[X^M_{n,t}]\,=\,[Y_n(t)]$.
\item \ $Y_n\,=\,\bigcup_{s\in\n} Y_n(s)$ \ is prefix free.
\item \ For any $n,s\in\n$, \,$[X_n]=[Y_n]$ \,and\, $[X^M_{n,s}] \subseteq [Y^N_{n,d(n+s)^2}]$.
\end{enumerate}
\end{clm}

%
\begin{proof} (Sketch)\hfill
\begin{enumerate}[label=\({\alph*}]
\item
Let us consider the following algorithm for the machine $N$: on input $(n,x)$
\begin{enumerate}[label=\({\arabic*}]
\item
 \ if there is \,$y\curl x$\, such that $y\in X^M_{n,|x|-1}$, then $N$ rejects the input (may loop indefinitely),
\item \ otherwise $\begin{array}[t]{l}\text{(2.1) \ if there is $y\curl x$ \ such that\  \,$y\in X^M_{n,|x|}$, then $N$ halts.}\\
\text{(2.2) \ otherwise, $N$ rejects $(n,x)$}.
\end{array}$
\end{enumerate}
If $N$ halts on an input, it does so in quadratic time. Let us set \ $Y\,=\,\text{dom}(N)$.
\item  One can check by induction on $s\in\n$, \,$[X^M_{n,s}]\,\subseteq\,[Y^N_n(s)]$.
\end{enumerate}
The rest follows.\end{proof}


\noindent Concerning the class \,\pspace\,, we note that if we replace \,$X^M_{n,t}$ in the above algorithm by \,$X^{M,\text{space}}_{n,t}$, the algorithm requires linear space. We thus get:

\begin{clm}
Under the same hypotheses as in the previous claim, we obtain \,$Y\subseteq\n\times\words$ \,and a constant $d\in\n$ such that for any $n,s\in\n$ \,and\, 
\begin{itemize}
\item \,$[X_n]=[Y_n]$, \, $Y_n$ prefix-free and
\item \,$[X_{n,s}^{M,\text{space}}] \subseteq [Y_{n,d(n+s)}^{N,\text{space}}]$. 
\end{itemize}
\end{clm}\newpage

\begin{rem} \label{rk strict increase}\hfill
\begin{enumerate}
\item As a consequence of these two claims, when considering our  complexity class \,\C, \,in the definition of ML-\,\C\,-S tests  $(G_n)_{n\in\n}$ (def~\ref{def-ML-test}), we shall assume the generating sets\linebreak $X_n\subseteq \words$, \,for $n\in\n$, to be prefix-free.
\item Let  us also note that if the controlling function  $f:\n\to\n$ is in \,\C, then the function \,$g:\n\to\n$ \ defined recursively by \ $\left\{\begin{array}{l}
\text{-}\ g(0)=f(0),\\
\text{-}\ g(n+1)\,=\,\max \{g(n)+1,f(n+1)\}
\end{array}\right.$
\\ is also in \,\C\ \ (for any $n\in\n,\ g(n)\leq\max\{f(m):m\leq n\}+n$) and since $f\leq g$, it also satisfies for any $n,i\in\n$,  $\mu([X_n])-\mu([X^M_{n,g(n+i)}])\leq 2^{-i}$. 
Hence we shall assume the controlling function to be strictly increasing.
\item This also applies to the controlling function in the definition of the  measure \,\C\ computable machine.
\end{enumerate}
\end{rem}

\noindent Following the notation in ~\cite[3.2.6]{nies}, we recall:

\begin{defi}\label{def-rmb}
Given a prefix-free Turing machine $M$ and $b\in\n$, one considers the subset $R^M_b$ of $\cantor$ defined as: 
\begin{align*} 
R^M_b =  [\{x\in\words: K_M(x)\leq |x|-b\}] \ 
      =  \{\alpha\in\cantor:\exists k\ K_M(\alpha\restriction k)\leq k-b\}.
\end{align*}
\end{defi}

\noindent We shall show:

\begin{prop}\label{prop-Ml-meas}\hfill
\begin{enumerate}[label=\({\alph*}]
\item Let \,\C\, be one of our time-complexity classes and let $(G_m)_{m\in\n}$ be an ML-\,\C\,-S test. Then there exists a measure \,\expo(\C)\,-computable machine $M$ such that

\centerline{\[\bigcap_{m\in\n}G_m\ \subseteq\ \bigcap_{b\in\n}R^M_b\]}
\item If $(G_m)_{m\in\n}$ is an ML-\,\pspace\,-S test, then there is a measure \,\pspace\ computable machine $M$ such that the above inclusion holds.
\end{enumerate}
\end{prop}

\noindent We shall push classical arguments as far as possible (see~\cite[3.5.18]{nies}). But a blind application of the Kraft-Chaitin theorem will not suffice  to produce the adequate  measure machine $M$ (we would only get $\Omega_M$ rightly approximable and we need more control than that). Hence after this exploration, we shall define a new goal and a strategy to reach it.

\begin{proof}\hfill
\begin{enumerate}[label=\({\alph*}]
\item \ Let \,\C\ \,be one of our time-complexity classes and let $(G_n)_{n\in\n}$ be an ML-\,\C\,-S test.\\ (By replacing $(G_n)_{n\in\n}$ \,by\, $(G_{2n})_{n\in\n}$, the controlling function $f$ by $n\mapsto f(2n)$, and using the inclusion \ $\bigcap_{n\in\n}G_n\ \subseteq\ \bigcap_{n\in\n}G_{2n}$) \,we  can assume  there exist \ $X\subseteq\n\times\words$ \ associated with a machine $N$, and a strictly increasing function $f$ in \,\C\, such that for any $n,i\in\n$,
\begin{enumerate}[label=\({\roman*}]\label{X}
\item  $X_n=\{x\in\words :(n,x)\in X\}=\{x\in\words : N(n,x)\downarrow\}\text{ \,is prefix-free}$
\item  $G_n =[X_n] \ \text{and} \ \mu(G_n)\leq 2^{-2n}$
\item $\mu([X_n])-\mu([X^N_{n,f(n+i)}])\leq 2^{-i}.$
\end{enumerate}


\noindent Let us consider the bounded request set 
\begin{equation}\label{request}
L\ =\ \{(|x|-m+1,x)\,:\,x\in X_m\}.
\end{equation}
Its weight \ $\gamma_L\,=\,\sum_{(m,x)\in X} 2^{-(|x|-m+1)}$ \ is \ $\leq 1$ \ by (ii).
\\ Let us apply the Kraft-Chaitin theorem (see~\cite[2.2.17]{nies},  ~\cite[3.6.1]{dohi}) to \,$L$\, (starting with an effective enumeration of $X$ and hence of $L$) to see how far it can bring us:
there is  machine $M$ such that the following properties hold:
\begin{equation}\label{prop-KC}\left\{\begin{array}{l}
\text{-} \  M \text{ \,is prefix-free},\\
\text{-} \ \Omega_M = \gamma_L, \\
\text{-} \ \text{for every }(m,x)\in X, \text{ there is a unique  } w(m,x)\in\words \text{ such that }\\ \makebox[.4cm]{} M(w(m,x))=x \text{ \,and\, } 
 |w(m,x)|=|x|-m+1,\\
\text{-} \ \text{dom}(M)\,=\,\{w(m,x):x\in X_m\}.
\end{array}\right.
\end{equation} 
\\Let us check that $\Omega_M$ can be approximated by a sequence \,$s:\n\to\q_2$ \,in\, \expo(\C), that is, for any $r\in\n$, \,one has \ $\Vert \Omega_M - s(r)\Vert \leq 2^{-r}$. This will not suffice, but it may give us some insight on how to replace \,$s(r)$ \,by  \,$\Omega_{M_{h(r)}}$, for $h$ in \,\expo(\C).
\\One has \ \ $\gamma_L\,=\,\Omega_M\,=\sum_{(m,x)\in X}2^{-(|x|-m+1)}\,=\,\sum_{m}\sum_{x\in X_{m}}2^{-(|x|-m+1)}$.
\\For $r\in\n$, let us set \ $\gamma_r\,=\,\sum_{m\leq r}\sum_{x\in X^N_{m,f(3r+1)}}2^{-(|x|-m+1)}$. 
 Let $r$ be fixed. 
\begin{itemize}
\item If $m\leq r$, then $f(m+2r+1)\leq f(3r+1)$. Hence for $m\leq r$

\espa $\mu([X_m\setminus X^N_{m,f(3r+1)}])\ =\ \mu([X_m])-\mu([X^N_{m,f(3r+1)}])\ \leq\ 2^{-(2r+1)}$.

The first equality holds because $X_m$ is prefix-free.
\item We deduce:
\begin{align}\label{bound}
\begin{split}
\gamma_L - \gamma_r &\ =\ \sum_m\sum_{x\in X_m}2^{-(|x|-m+1)}\ -\ \sum_{m\leq r}\ \sum_{x\in X_{m,f(3r+1)}}2^{-(|x|-m+1)}\\
 &\ \leq\ \sum_{m\leq r}\ \ \sum_{x\in X_m\setminus X_{m,f(3r+1)}}2^{-(|x|-m+1)}\ +\ \sum_{m> r}\ \sum_{x\in X_m}2^{-(|x|-m+1)}\\
 &\ \leq\ \sum_{m\leq r}2^{m-1}2^{-(2r+1)}\ +\ \sum_{m> r}2^{m-1}\mu(G_m)\\
 &\ \leq\ 2^{-r}.
 \end{split}
 \end{align}
 \end{itemize}
 
\noindent We note that the function \ $\varphi:\begin{array}[t]{l}
 \n\to\q_2 \vspace{-1mm}\\
 \ r\mapsto\gamma_r
 \end{array}$
 \ belongs to \,\expo(\C) \,(for all $z$ such that \linebreak$|z|\leq f(3r+1)$, we have to check whether \ $z\in X^N_{m,f(3r+1)}$).
\\ Also by the properties of~\eqref{prop-KC}, for $r\in\n$, one has:
 \begin{align}\label{gamma-r}
 \gamma_r\ =\ \sum_{m\leq r}\ \sum_{x\in X_{m,f(3r+1)}}2^{-|w(m,x)|}.
\end{align}
\\This suggests the following claim (inside the proof of Proposition~\ref{prop-Ml-meas}, we use a local numbering with letters of definitions and claims):

\setcounter{cl}{0}

\begin{cl}\label{goal-fact}
Let $M$ be a machine satisfying the properties of ~\eqref{prop-KC} and let $h:\n\to\n$ be a function  such that for any $r\in\n$, \ 
$\big\{w(m,x):m\leq r,\ x\in X^N_{m,f(3r+1)}\big\}\ \subseteq\ \text{dom}(M_{h(r)})$.\\
Then for any $r\in\n$, \ $\Omega_M - \Omega_{M_{h(r)}}\ \leq \ 2^{-r}$.
\end{cl}

\begin{proof} For $r\in\n$, the inclusion \ $\big\{w(m,x):m\leq r,\ x\in X^N_{m,f(3r+1)}\big\}\,\subseteq\, \text{dom}(M_{h(r)})$ and (\!~\ref{gamma-r}) \ imply \\
 \centerline{$\gamma_r\,\leq\,\mu([\text{dom}(M_{h(r)})])\,\leq\,\Omega_M\,=\,\gamma_L$.}

Hence by (\!~\ref{bound}), \ $\Omega_M -\Omega_{M_{h(r)}}\,\leq\, \gamma_L -\gamma_r\,\leq\, 2^{-r}$.
\end{proof}

\textbf{The strategy to obtain Proposition~\ref{prop-Ml-meas}(a)} (induced by the previous claim): \\
 \textbf{(a.1)} We define a (c.e) well-ordering  \ $\squ _X$ \ on \,$X$.\\
 \textbf{(a.2)} We then develop a proof of Kraft-Chaitin theorem by  \ $\squ _X$-induction, in order to obtain a function \ $w:X\to\words$ \ and a machine $M$ as in \eqref{prop-KC} (sending \,$w(m,x)$ \,to\, $x$).\\
\esp \textbf{(a.3)} A sufficient condition to satisfy the hypotheses of Claim~\ref{goal-fact}, is  that $M$ works in  \,\expo(\C)\, time (on successful computations). To define such an algorithm for $M$, the idea is to stratify $X$ as \,$\bigcup_{r\in\n}X(r)$\, so that $X(r)$ is a finite initial segment of $X$ for $\squ _X$ and such that for any\linebreak  $(m,x)\in X$, if $|w(m,x)|\leq r,$ then  \,$(m,x)\in X(r)$.
Hence for any $w\in\words$ such that $|w|\leq r$, to check whether \,$w=w(m,x)$\, for some $(m,x)\in X$ (and thus to send \,$w$\, to $x$), we shall only have to carry out the $\squ _X$-induction process on $X(r)$.\\
 \textbf{(a.4)} We finally define  an algorithm for $M$ based on the previous remark which satisfies the time bounds.

\textbf{(a.1) The definition of the well-ordering: }

\begin{df}\hfill
\begin{itemize}
\item Let \ $\squ$ \ be the ordering on $\n\times\words$ defined as follows: for $m,m'\in\n,\ x,x'\in\words$

$(m,x)\squ (m',x')\ssi \begin{cases}
x\curlex x'\ \text{ or}\\
x=x' \text{ \ and \ } m\leq m'.
\end{cases}$
\item Let \ $\squ _X$ \ denote the restriction of \ $\squ$ \ on $X$. 
\end{itemize}
\end{df}

\begin{cl}\label{pred-fini}
For any $(m,x)\in X$, the set \,$\{(m',x')\in X:(m',x')\squ_X (m,x)\}$ \ is finite.
 Hence every element of $X$ - except the least one - admits an immediate predecessor for \ $\squ_X$.
\end{cl}

\begin{proof} \ 
Let $(m,x),(m',x')\in X$ \,be such that \,$(m',x')\squ (m,x)$.
\\Necessarily \ $|x'|\leq |x|$. Since \,$x'\in X_{m'}$, we have \ $[x']\subseteq [X_{m'}]$. \ Hence \ $2^{-|x'|}\leq\mu([X_{m'}])\leq 2^{-2m'}$.
\\Therefore \ $2m'\leq|x'|\leq|x|$. 
\end{proof}

\noindent We then set:

\begin{df}\hfill
\begin{itemize}
\item For $(m,x)\in\n\times\words$, let \,$r_{m,x}=|x|-m+1$.
\item For $(m,x)\in X$, \,let \,$\text{pred}(m,x)$ \,be  the immediate predecessor of $(m,x)$ for \,$\squ_X$; if $(m_0,x_0)$ is the least element of $X$ for \,$\squ_X$, we set \,$\text{pred}(m_0,x_0)\,=\,(\emptyset,-1)$.
\item For a set $Z$, let \,$\+P_{\text{finite}}(Z)$ \,be the collection of finite subsets of $Z$.
\end{itemize}
\end{df}

\textbf{(a.2) The inductive construction (~\cite[2.2.17]{nies} of Kraft-Chaitin Theorem based on \,$\squ_X$:}\\
We define by induction on \,$\squ_X$\, the following functions:

\espa$\left(\begin{array}{ccl}
R\ :&X\ \to&\+P_{\text{finite}}(\words)\\
w\ :&X\ \to&\words\\
z\ :&X\ \to&\words.
\end{array}\right.$

\begin{itemize}[label=$-$]
\item \,$R(\emptyset,-1)=\{\emptyset\}$.
\item At step \,$(m,x)\in X$, we suppose \,$R(\text{pred}(m,x))$ is known and define \,$z(m,x),\,w(m,x)$ and $R(m,x)$:
\begin{enumerate}[label=\({\roman*}]
\item Let $z(m,x)$ be the longest string in $R(\text{pred}(m,x))$ of length $\leq\,r_{m,x}$.
\item Let $w(m,x)$ be the leftmost string (least for lexicographic order) of length $r_{m,x}$ extending $z(m,x)$ (i.e. \,$w(m,x)=z(m,x)0^{r_{m,x}-|z(m,x)|}$)
\item One sets \,$R(m,x)\,=\,(R(\text{pred}(m,x))\setminus \{z(m,x)\})\,\bigcup\,\{z(m,x)0^i1:0\leq i<r_{m,x}-|z(m,x)|\}$.
\end{enumerate}
\end{itemize}

\noindent By classical arguments~\cite[2.2.17]{nies}, the construction can be carried out and is effective (by the proof of Claim~\ref{pred-fini} and by Claim~\ref{algo}, there exists an effective increasing - with respect to $\squ_X$ - enumeration of $X$). Hence with $w:X\to\words$ defined as above, there is a machine $M$ satisfying Properties~\eqref{prop-KC} which for $(m,x)\in X$ sends $w(m,x)$ to $x$.\\
Instead of justifying the assertion about the existence of an effective increasing enumeration of $X$, let us define now the stratification of $X$ which will allow us to exhibit an algorithm for $M$ with the appropriate time bounds.\\

\textbf{(a.3) The stratification of $X$:}

\begin{df}\label{X(r)}
For $r\in\n$, let \ $X(r)\ =\ \big\{(m,x):x\in X^N_{m,f(3r+1)},\ |x|\leq 2r,\ m\leq r\big\}$.
\end{df}

\noindent The last requirement ``$m\leq r$" is redundant, we left it to stress the fact that $X(r)$ is finite.
 Let us note the following:

\begin{cl}\label{algo}  If $(m,x)\in X$, then
\begin{enumerate}[label=\({\alph*}]
\item  $|x|\leq 2r_{m,x},\ 2m\leq |x|$ \,and\, $m\leq r_{m,x}$,
\item $x\in X^N_{m,f(3_{r_{m,x}}+1)}$.
\end{enumerate}
\end{cl}

\begin{proof}  \ Let $(m,x)\in X$.
\begin{enumerate}[label=\({\alph*}]
\item As we noted in the proof of Claim~\ref{pred-fini}, $x\in X_m$ implies $|x|\geq 2m$. Hence we deduce\linebreak \,$r_{m,x}=|x|-m+1\geq m+1$ \,and\, $r_{m,x}\geq |x|-(|x|/2)+1\geq |x|/2$.
\item Since \ $\mu([X_m])-\mu([X^N_{m,f(m+|x|+1)}])\leq 2^{-(|x|+1)}$, \,necessarily \,$x\in X^N_{m,f(m+|x|+1)}$. \ 
Now by (a), $m\leq r_{m,x},\ |x|\leq 2r_{m,x}$, hence since $f$ increasing, \,$x\in X^N_{m,f(3r_{m,x}+1)}$.\qedhere
\end{enumerate}
\end{proof}

\noindent The interest of the stratification of $X$ appears in the following:

\begin{cl}\label{lemma-stratX}\hfill
\begin{enumerate}[label=\({\alph*}]
\item \ Let $r\in\n$. Then $(X(r),\squ_X)$ \ is a finite initial segment of $(X,\squ )$.
\item \ If $r\in\n$,  $(m,x)\in X$ \ and \ $r_{m,x}\leq r$, then \,$(m,x)\in X(r)$.
\end{enumerate}
\end{cl}

\begin{proof}\hfill \begin{enumerate}[label=\({\alph*}]
\item
 Let $(m',x'),(m,x)$ \ be both in $X$, \ $(m',x')\squ (m,x)$ \ and \ $(m,x)\in X(r)$.\\
 Necessarily \ $|x'|\leq |x|.$ \ Hence $|x'|\leq 2r$. Also 
\[(m',x')\in X\ \Rightarrow\ \left(\begin{array}{l}
 2m'\leq |x'|\\
 x'\in X_{m',f(m'+|x'|+1)}
 \end{array}\right.
\]
\\ We deduce \ $m'\leq r$ \ and \ $x'\in X_{m',f(3r+1)}$. Therefore \ $(m',x')\in X(r)$.
 \item Let \ $(m,x)\in X$ \ and \ $r_{m,x}\leq r$. Then by Claim~\ref{algo}(b), \,$x\in X^N_{m,f(3r+1)}$.\qedhere
 \end{enumerate}
 \end{proof}
 
 \textbf{(a.4) The algorithm for $M$:}
\\ The proof of the Kraft-Chaitin theorem yields the following:
 
 \begin{cl}\label{claim-R-encoded}
 Let $r\geq 1$. \begin{enumerate}[label=\({\alph*}]
\item
  For any  $(m,x)\in X(r),\ \,R(m,x)$ \,contains at most 2r+2 strings which are of length\ $\leq\,2r+1$. \ Hence there is a constant $d\in\n$ such that \,$R(m,x)$ can be coded (according to increasing length) by a string \,$\rho(m,x)$ \,of length $\leq dr^2$.
 \item Also if $(m',x')=\text{pred}(m,x)$ with $(m',x'),\,(m,x)\in X(r)$, there is an algorithm requiring $\+{O}(r^2)$ steps which produces \,$\rho(m,x)$\, from \,$\rho(m',x')$.
 \end{enumerate}
 \end{cl}
 
 \proof \ We use the classical fact that all strings  have different length, adding the bound information.
\begin{enumerate}[label=\({\alph*}]
\item
\ Let $r\in\n$ be fixed. One argues by \,$\squ_X$\,- induction:  let $(m,x)\in X(r)$. We assume that all strings in $R(\text{pred}(m,x))$ \ have different length $\leq\ 2r+1$.

By construction (requirements (i),(iii)), we obtain:\begin{itemize}
\item\ $R(m,x)\setminus R(\text{pred}(m,x))\ \subseteq\ \big\{z\in\words:|z(m,x)|< |z|\leq r_{m,x}\big\}$,
\item \ $R(\text{pred}(m,x))\,\cap\,\big\{z\in\words:|z(m,x)|<|z|\leq r_{m,x}\big\}\ =\ \emptyset$,
\item \ All strings in \ $R(m,x)\setminus R(\text{pred}(m,x))$ \ have different length $\leq \ r_{m,x}$.
\end{itemize}
Now \ $r_{m,x} = |x|-m+1\ \leq\ |x|+1\ \leq\ 2r+1$.

Hence by induction hypothesis, all elements in $R(m,x)$ have different length $\leq\ 2r+1$.
\item \ Also if the elements of \,$R(\text{pred}(m,x))$\, are enumerated according to increasing length, it takes $\+{O}(r^2)$ steps to build \,$R(m,x)$\, enumerating its elements according to increasing length.\qedhere\end{enumerate}

 
\noindent By Claim~\ref{lemma-stratX},  we can replace the induction on $\ \squ_X$ by an induction on \ $\squ_{X(r)}$:

 \begin{theo}\label{algo-below}\hfill
 \begin{itemize}
 \item \ For $r\in\n$, let \ $A(r)\ =\ \big\{(m,x)\in\n\times\words\,:\,m\leq r,\ |x|\leq 2r\big\}$. We shall enumerate \,$A(r)$\, according to \,$\squ$. Let $\rho_0$ code the set $\{\emptyset\}$; $\rho$ will be a variable whose value is $\rho(m,x)$ where $(m,x)$ is the last element of $X(r)$ which has been treated. 


\item \  On input $w\in\words$ such that $|w|=r$,
\begin{itemize}[label=$-$]
\item
 set $\rho:=\rho_0$,
\item $(m,x)\in A(r)$.

 \emph{Case 1}: if $x\in X^N_{m,f(3r+1)}$, then from $\rho$ (playing the role of $R(pred(m,x))$), compute the values $w(m,x)$ and $\rho(m,x)$, and set $\rho:=\rho(m,x)$.\\
\esp\emph{Case 1.1}: if  $w(m,x)=w$, then output $x$ and stop the machine.\\
\esp\emph{Case 1.2}: otherwise \\
\esp\esp\emph{Case 1.2.1}: if $(m,x)=(r,1^{2r-1})$ (the maximum of  $A(r)$), then loop for ever,\\
\esp\esp\emph{Case 1.2.2}: otherwise compute the successor $(m',x')$ of $(m,x)$ for $\squ $ in $A(r)$ and go to step $(m',x')$.\\
\emph{Case 2}: if $x\notin X^N_{m,f(3r+1)}$, \\
\esp\emph{Case 2.1}: if $(m,x)=(r,1^{2r-1})$, then loop for ever,\\
\esp\emph{Case 2.2}: otherwise, compute  the successor $(m',x')$ of $(m,x)$ for $\squ $ in $A(r)$ and go to step $(m',x')$.
\end{itemize}\end{itemize}
 \end{theo}

 \noindent Let \,\C\ =\ $\bigcup_{g\in F_C}\text{FDTIME}(g(n))$ be one of our time-complexity classes.
 Since the integers are under unary representation, $f(n)=\+O(f'(n))$
 for some $f'\in F_C$. Using Claim~\ref{claim-R-encoded}, one can then
 show the existence of a constant $c_0$ and of a function $g\in F_C$
 so that, for any $|w|=r$, any $(m,x)\in A(r)$, if step $(m,x)$ is
 finite, it takes at most$\ c_0 g(r)$ steps to be completed.

 Now, for $r\in\n$, $|A(r)|\leq 2^{2r+1}(r+1)$.

 \textbf{The time bound:}
 hence there exists $l\in$ \C\ such that any successful computation of $M$ on a finite sequence $w$ with $|w|=r$, 
 takes at most $l(r)2^{3r}$ steps. Therefore
 \begin{align}
 \label{near-goal}
 \text{dom}(M)\cap \{w\in\words:|w|\leq r\}\ &=\ \big\{w(m,x):(m,x)\in X \text{ and } r_{m,x}\leq r\big\}\nonumber\\
   &\subseteq\ \ \text{dom}(M_{l(r)2^{3r}}).
   \end{align}
 Our goal is to fulfill the conditions of Claim~\ref{goal-fact}, that is to obtain a function $h$ in \ \expo(\C)\ \,so that:
\begin{align*}
\big\{w(m,x):m\leq r,\ x\in X_{m,f(3r+1)}\big\}\ \subseteq\ \text{dom}(M_{h(r)}).
\end{align*}
\\If $x\in X_{m,f(3r+1)}$, then $|x|\leq f(3r+1)$. Hence  $r_{m,x}=|x|-m+1\leq f(3r+1)+1$. Therefore
\\$\begin{array}[t]{ccl}
\big\{w(m,x):m\leq r,\ x\in X^N_{m,f(3r+1)}\big\}\!\!\!\!&\!\subseteq\!\!\!&\!\!\!\big\{w(m,x)\!:\!(m,x)\!\in\! X \text{ and } r_{m,x}\!\leq \!f(3r+1)\!+\!1\big\}\\
\!\!\!\!\!\!\!\! &\subseteq\!\!\!& \! \text{dom}(M_{l(f(3r+1)+1)2^{3(f(3r+1)+1)}}\big\}\text{\esp (by (\!~\ref{near-goal}))}.
\end{array}$
 \\ Let us define $h$ by $h(r)=l(f(3r+1)+1)2^{3(f(3r+1)+1)}$,  for $r\in\n$. Then $h$ belongs to \ \expo(\C), and by Claim~\ref{goal-fact}, \ $\Omega_M -\Omega_{M_{h(r)}}\,\leq\,2^{-r}$. Hence $M$ is a measure\ \expo(\C) \ computable machine.

 One finally concludes the proof of Proposition~\ref{prop-Ml-meas}(a) by classical arguments: \\Let $m\in\n$
. \ If $x\in X_m$, then \ $M(w(m,x))=x$. \,Hence \ $K_M(x)\leq |w(m,x)| = |x|-m+1$.\\
 Therefore, for $m\geq 1,\ \,G_m\,=\,[X_m]\,\subseteq\,R_{m-1}^M$.\\


\item (b) \ Let now \ \C\ = \ \pspace.
\\We assume  $([X_n])_{n\in\n}$ is an ML-\ \pspace\ -S test \,(defined from a set $X$ associated with a machine $N$). By Remark~\ref{controlling}, let $f$ be a polynomial function such  that for any $n,i\in\n$,
\esp $X_n$ is prefix-free,
\ $
\mu([X_n])\,\leq \,2^{-2n}$ \ and \ 
$\mu([X_n]) - \mu([X_{n,f(n+i)}^{N,\text{space}}])\,\leq 2^{-i}$.

Let us start from the bounded request set defined as above, our goal is to define a machine $M$ satisfying properties (\ref{prop-KC}) and a polynomial function $h$ such that, for any $r\in\n$,
\begin{align}\label{inc-pspace}
\big\{w(m,x):m\leq r,\ x\in X^{N,\text{space}}_{m,f(3r+1)}\big\}\ \subseteq\ \text{dom}(M^{space}_{h(r)}).
\end{align}
If we replace \,$X^N_{m,t}$ \,by \,$X_{m,t}^{N,\text{space}}$ \,and\, $M_t$ \,by\, $M^{\text{space}}_t$ \,in the previous definitions and claims, the transposed definitions and claims \,remain valid. Now we must define an algorithm for $M$  with ``$x\in X_{m,f(3r+1)}^{N,\text{space}}$" in place of \ ``$x\in X^N_{m,f(3r+1)}$".\\
- We use the fact that there is a constant $d\in\n$ such that for any \,$m,t\in\n$, \ $X_{m,t}^{N,\text{space}}\ \subseteq\ X^N_{m,2^{d t}}$. Hence using $2^{df(3r+1)}$ (under binary representation) as a time-counter, we can check in space $k(f(3r+1))$, for some constant $k$, \,whether \,``$x\in X^{N,\text{space}}_{m,f(3r+1)}$".\\
- Also we note that in the previous algorithm, in checking successively for all $(m,x)\in A(r)$ whether $w(m,x)=w$, we needed only to keep track of the value $\rho(m,x)$, for the last browsed $(m,x)$.

Hence we can define an algorithm for $M$ such that for some  polynomial function $g$,  any successful computation of $M$ on $w$ with $|w|=r$ requires at most $g(r)$ cells.  Let us set, for $r\in\n$, \ $h(r)=g(f(3r+1)+1)$, \,$h$ satisfies (\!~\ref{inc-pspace}).

One concludes by the  equivalent of Claim~\ref{goal-fact} for space, that $M$ is a measure \pspace\ computable machine, and as above that for $m\geq 1,\ \,G_m\,=\,[X_m]\,\subseteq\,R_{m-1}^M$.
\end{enumerate}
\noindent This concludes the proof of Proposition~\ref{prop-Ml-meas}.
\qed
\end{proof}

 One obtains the opposite direction as an easy generalization of the classical case.

\begin{prop}\label{prop-comp}\hfill
\begin{enumerate}[label=\({\alph*}]
\item
 Let \ \C\ \ be one of our time-complexity classes. If $M$ is a measure  \C\ computable machine, then \ $(R_b^M)_{b\in\n}$ \ (Definition~\ref{def-rmb}) \ is an ML-\ \expo(\C)\,-S test.
\item If $M$ is a measure \,\pspace\ \,computable machine, then \ $(R_b^M)_{b\in\n}$ \ is an ML-\ \pspace\ -S test. 
\end{enumerate}
\end{prop}

\begin{proof} \ We refer here to ~\cite[3.5.14,\,3.5.18]{nies}.
\hfill
\begin{enumerate}[label=\({\alph*}]
\item Let $M$ be a prefix-free machine and let $g\in\ $\C\ \,be strictly increasing and such that, for any $i\in\n$, \ $\Omega_M -\Omega_{M_{g(i)}}\leq 2^{-i}$.
As in~\cite[3.5.15]{nies}, one shows:

\begin{clm}\label{KM}
For any $x\in\words,\ b\in\n$, \ $
K_M(x)\leq |x|-b\ \Leftrightarrow\ K_{M_{g(|x|)}}(x)\leq |x|-b.$
\end{clm}

\noindent Let \ $X\ =\ \big\{(b,x)\in\n\times\words: K_M(x)\leq |x|-b\big\}$. We consider the machine $N$ which on input $(b,x)$ tests for each $y\in\words$ such  that $|y|\leq |x|-b$ whether \ $M_{g(|x|)}(y)=x$. If there is such a $y$, $N$ halts, otherwise it diverges.\\
Then by Claim~\ref{KM}, \ $N(b,x)\downarrow\ \Leftrightarrow\ (b,x)\in X$. For $b\in\n,$ one has \ $R^M_b\,=\,[X_b]$ \ and \ classically \ $\mu([X_b])\leq 2^{-b}$.

\begin{clm}\label{Xb}
$([X_b])_{b\in\n}$ \,is an ML-\,\expo(\C)-S test
\end{clm}

\proof (Sketch refering to~\cite[3.5.18]{nies}).\\
By definition of the machine $N$, there exists  an increasing function $f$ in\ \,\C\ \,such that if $N$ halts on $(b,x)$, it does so in at most \ $2^{|x|}f(|x|+b)$ \ steps. Hence\\ 
\centerline{$X_b\,\cap\,\big\{x\in\words:|x|\leq m\big\}\ \subseteq\ X^N_{b,2^mf(m+b)}$.}
\\Let \ $h\in$\ \expo(\C) \ be such that \ $h(r)=2^{g(r)}f(g(r)+r)$. Then

\centerline{$X_b\,\cap\,\big\{x\in\words:|x|\leq g(m)\big\}\ \subseteq\ X^N_{b,h(m+b)}$.}
Since \ $M(\sigma)=x$ \,and\, $|x|>g(m)$ \,imply \,$\sigma\notin\text{dom}(M_{g(m)})$, one  then argues classically to deduce
$$\mu([X_b])-\mu([X^N_{b,h(m+b)}])\ \leq\ 2^{-b}(\Omega_M
  -\Omega_{M_{g(m)}})\ \leq\ 2^{-m}\,.\eqno{\qEd}
$$\smallskip


\item Let now \,\C \ be \,\pspace . We suppose $M$ is a prefix-free machine, $g$ is a polynomial function  such that for any $i\in\n ,\ \Omega_M -\Omega_{M_{g(i)}^{\text{space}}}\ \leq\ 2^{-i}$.

The set $X$ is defined as in the time-complexity case, but in the rest of the argument, we replace \,$M_t$\, by \,$M^{\text{space}}_t$. Claim~\ref{KM} can be transposed. On input $(b,x)$, using a time-counter as previously, the machine $N$ tests  in space $p(|x|+b)$, for a polynomial function $p$, whether there exists $|y|\leq|x|-b$ such that \ $M_{g(|x|)}^{\text{space}}(y)=x$. Setting $h(r)=p(g(r)+r)$, 
one concludes as above that \ $\mu([X_b])-\mu([X_{b,h(b+m)}^{N,\text{space}}])\ \leq\ 2^{-m}$. This concludes the proof of Proposition~\ref{prop-comp}.\qedhere
\end{enumerate}
\end{proof}

\noindent From Propositions~\ref{prop-Ml-meas} and~\ref{prop-comp}, we deduce:

\begin{thm}\label{thm1}\hfill
\begin{enumerate}[label=\({\alph*}]
\item
 Let \,\C\ \,be one of our time complexity classes. Then for any $\xi\in\cantor$,
 \begin{itemize}
\item $\xi$ is ML-\,\expo(\C)\,-S random \ $\Rightarrow$\ \ $\xi$ is Kolmogorov-\,\C\,-S random.
\item $\xi$ is Kolmogorov-\,\expo(\C)\,-S random \ $\Rightarrow$\ \ $\xi$ is ML-\,\C\,-S random.
\end{itemize}
\item Let \,\C\ be the class \,\pspace,\ \toexp\ \,or \,\primrec. Then for any $\xi\in\cantor$,
 \ $\xi$ is ML-\,\C\,-S random $\Leftrightarrow$\ \ $\xi$ is Kolmogorov-\,\C\,-S random.
\end{enumerate}
\end{thm}

\subsection{The relation between the Martin-L\"of and the martingale notions.} 

Let us deal now with the notion of randomness associated with martingales and orders. We refer to the notion of ``inverse" given in Definition~\ref{def-true-order}. As a way to obtain true \,\C\ \,orders, let us note:

\begin{clm}\label{inv-inv}\hfill
\begin{enumerate}[label=\({\alph*}]
\item
 If $f$ is an order, then $\text{Inv}_f$ is also an order and for $i\geq 1,\ \text{Inv}_{\text{Inv}_f}(i)=f(i-1)+1$.
\item If $f$ is a strictly increasing function in the class \,\C\,, where \,\C\ \,is one of our complexity classes, then \,$\text{Inv}_f$ \,is a true \,\C\ \,order.
\end{enumerate}
\end{clm}

\begin{proof}\hfill
\begin{enumerate}[label=\({\alph*}] 
\item Let $f$ be an order. Then for $n\in\n$,
$\text{Inv}_f(n+1)\,=\,\big(\text{least }k\ f(k)\geq n+1\big)\, \geq\,\big(\text{least }k\ f(k)\geq n)\big)\, = \,\text{Inv}_f(n)$.
\\For $i\in\n,\ f(i)<f(i)+1$, hence $\text{Inv}_f(f(i)+1)\,>\,i$. Therefore \,$\text{Inv}_f$ is unbounded.\\
We deduce that $\text{Inv}_f$ is an order.
\\Now if \ $h=\text{Inv}_f$, \,let us compute \ $\text{Inv}_h(i)$, \,for $i\geq 1$. We have the equivalences:
\begin{align*}
h(k)\geq i\ &\Leftrightarrow\ (\text{least }n\  f(n)\geq k) \geq i\\
 &\Leftrightarrow\ f(i-1)<k.
\end{align*}
Therefore, if $i\geq 1$, then \ $\text{Inv}_h(i)=f(i-1)+1$.\\

\item Let $f$ be strictly increasing in \,\C. We deduce that  for $n\in\n,\ f(n)\geq n$. Hence \ $\text{Inv}_f(n)=\text{least }k\leq n\ f(k)\geq n$. This implies that given our choice of classes \,\C, \ $\text{Inv}_f$ is also in \,\C. \ By (a), $\text{Inv}_f$ is an order, and \,$ \text{Inv}_{\text{Inv}_f}$ is in \,\C. Hence $\text{Inv}_f$ is a true \,\C\ order.\qedhere
\end{enumerate}\end{proof}

\noindent We propose now a result which will be useful in the next section. As an immediate consequence, it shows that when \,\C\ \,is \,\pspace, requiring the functions from $\n$ to $\n$  - especially orders - to be in \pol\ or in \pspace\ yields the same notion of randomness.

\begin{clm}\label{bound-inv}\hfill
\begin{enumerate}[label=\({\alph*}]
\item
 Let \,\C\, be one of our classes and let $f$ be a true\, \C\ order. Then there exists a strictly increasing function $g$ in \,\C\, such that \ $\text{Inv}_g(n)\leq f(n)$ \ a.e.
\item Let now $f$ be a true \,\pspace\ \,order. Then there is a true \,\pol\ \,order $h$ such that for any $n\in\n, \ \  h(n)\leq f(n)$.
\end{enumerate}
\end{clm}

\begin{proof}\hfill
\begin{enumerate}[label=\({\alph*}]
\item \,Let $f$ be a true \,\C\, order. We consider the function $f'$ defined by \ $f'(i)=f(i+1)\dotminus 1$. \,Then $f'$ is an order in \,\C. \,Also for $n>0$, we have:

\centerline{$\text{Inv}_{f'}(n)\,=\,\text{least }k\ (f(k+1)-1\geq n)\,=\,\text{Inv}_f(n+1)\dotminus 1$.}
Hence $\text{Inv}_{f'}$ is a true \,\C\, order.
\\Let us define inductively the function $g$:
\\$\left( \begin{array}{l}
g(0)=\text{inv}_{f'}(0)=0\\
g(n+1)=\max \{\text{Inv}_{f'}(n+1),g(n)+1\}.
\end{array}\right.$
\\ For all our classes \,\C, \,$g$ is in \,\C, \,it is strictly increasing and satisfies \,$g\geq \text{Inv}_{f'}$. Hence by definition of $\text{Inv}, \ \text{Inv}_g\leq \text{Inv}_{\text{Inv}_{f'}}$.
\\Now by Claim~\ref{inv-inv}, for $i\geq 1, \ \text{Inv}_g(i)\,\leq\, f'(i-1)+1\,=\,(f(i)\dotminus 1)+1$.
\\Let $i_0$ be least such that $i_0\geq 1$ and \,$f(i_0)>0$. Then for any $i\geq i_0,\ \,\text{Inv}_g(i)\leq f(i)$.
\\

\item Let $f$ be a true \,\pspace\, order.\\ We define $f'$ and $g$ from $f$ as in (a). The function \,$g$\, is thus in \,\pspace\, and \,$g\geq\text{Inv}_{f'}$. Now since integers are under unary representation, there must exist $d,b,k\in\n^*$ such that for any $n\in\n,\ g(n)\leq dn^k+b$.
Let $p(n)=dn^k+b$. Then \,$p$\, is strictly increasing and $\,p\geq g\geq \text{Inv}_{f'}$.
We deduce as above \ $\text{Inv}_p(n)\leq f(n)$ \ a.e.
\\By Claim~\ref{inv-inv}(b), \,$\text{Inv}_p$ \,is a true \,\pol\, order.\qedhere
\end{enumerate}\end{proof}

\noindent Following the terminology of~\cite{lu1} in the polynomial time context, we set:

\begin{defi}\label{approx}
Let $D$ be some class of functions. A martingale $d:\words\to\reel^+$ \,is $D$-approximable if there exists $F:\words\times\n\to\q_2$ \,in $D$ such that for any $i\in\n,\ \Vert d(x)-F(x,i)\Vert\,\leq\, 2^{-i}$.
\end{defi}

\noindent The following type of results has already been obtained (~\cite{amboter,julu,mayo}...).
\begin{lem}\label{q2appro}
Let $g$ be some time-constructible function. If \,$V:\words\to\reel^+$ is a martingale which is $\text{FDTIME(g(n))}$-approximable, then there exists a $\q_2$-valued martingale $d$ in\linebreak $\text{FDTIME(ng(2n+4))}$ \,such that for any $x\in\words,\ V(x)\leq d(x)\leq V(x)+2$.
\end{lem}

\noindent We omit the argument (one can adapt~\cite[7.3.8]{nies} or look up the above references).
\\To study  the relations between ML tests and the martingale\,-\,order conditions in the Schnorr subrecursive framework, we shall resort to the following notion from measure theory:

\begin{defi}\label{meas-condi}
Given a measurable subset $A$ of $\cantor$ and $x\in\words$, the conditional measure \,$\mu(A|x)$ is the quotient \,$\frac{\mu(A\cap[x])}{\mu([x])}\,=\,2^{|x|}\mu(A\cap[x])$.
\end{defi}

\noindent Classically, the function \,$d:x\mapsto\mu(A|x)$ \,is a martingale. 

\begin{prop}\label{ml-mart}\hfill
\begin{enumerate}[label=\({\alph*}]
\item Let \,\C\ \,be one of our time-complexity classes. If \,$(G_n)_{n\in\n}$\, is an ML-\,\C\,-S test, then there exist an \,\expo(\C)\,-approximable martingale $B$ and a true \,\C\, order $h$ such that

\centerline{for any $\,\xi\in\cantor $, \  $\xi\in\bigcap_{n\in\n}G_n\ \Rightarrow\ B(\xi\restriction i)\geq 2^{h(i)}$ \ \ i.o.}
\item Given an ML-\,\pspace\,-S test \,$(G_n)_{n\in\n}$,\, there exist a \,\pspace\,-approximable martingale $B$ and a true \,\pspace\, order $h$ satisfying the above implication for any $\xi\in\cantor$.
\end{enumerate}
\end{prop}

\begin{proof}\hfill
\begin{enumerate}[label=\({\alph*}]
\item
 \ Let \,\C\, be one of our time-complexity class and let $(G_n)_{n\in\n}$ be an ML-\,\C\,-S test. We can assume there are \,$X\subseteq\n\times\words$, \,a machine $M$ and a function $f:\n\to\n$ strictly increasing in \,\C\, such that
\begin{itemize}[label=$-$]
\item\ $X\,=\,\big\{(n,x)\in\n\times\words:\,M(n,x)\downarrow\big\},$
\item $\text{ for }n\in\n,\ \,X_n\,=\,\big\{x\in\words:(n,x)\in X\big\} \text{ \,is prefix-free,} G_n=[X_n] \text{ and } \mu(G_n)\leq 2^{-2n}$,
\item$\text{  for }n,i\in\n,\ \mu([X_n])-\mu([X^M_{n,f(n+i)}])\,\leq\, 2^{-i}.$
\end{itemize}

\begin{defi}\label{def-mart}\hfill
\begin{itemize}
\item let $g:\n\to\n$ \ be such that, for $i\in\n ,\ g(i)=f(5i)$.
\item For $n,k\in\n$, \,let \,$C^k_n\,=\,[X_n\setminus X^M_{n,g(k)}]$.
\item Define \,$B:\words\to\reel ^+$ as follows:\, for $x\in\words, \ \,B(x)\,=\,\sum_{n,k} 2^k\mu(C^k_n|x)$.
\end{itemize}
\end{defi}

\noindent For $n,k\in\n$, since $X_n$ is prefix-free, we have \,$[X_n\setminus X^M_{n,g(k)}]\,=\,[X_n]\setminus[X^M_{n,g(k)}]$. Let us note that for $n\leq 2k,$
\begin{align}
 \mu(C^k_n)\,=\,\mu([X_n])-\mu([X^M_{n,g(k)}])\, \leq\, \mu([X_n])-\mu([X^M_{n,f(n+3k)}])\, \leq \, 2^{-3k}.\label{bound-below-2k}
 \end{align}
Hence for any $n\in\n$,\begin{align}
 \mu(C^k_n)\ \leq\ \mu(G_n)\ \leq \ 2^{-2n}\label{bound-2n}.
\end{align}

\begin{clm}\label{Bmart}
$B$ is a martingale.
\end{clm}

\proof \esp We only need to check that $B(\emptyset)$ is finite. One has
\begin{align*}
B(\emptyset) &= \sum_k\sum_{n\leq 2k} 2^k\mu(C^k_n) +\sum_k\sum_{n>2k}2^k\mu(C^k_n)\\
&\leq \sum_k\sum_{n\leq 2k} 2^k2^{-3k} +\sum_k\sum_{n>2k}2^k2^{-n} \text{\esp\ (by(~\ref{bound-below-2k}) and\,(\!~\ref{bound-2n}))}\\
&\leq \sum_k 2k2^{-2k} + \sum_k 2^k2^{-2k}\\
&\leq 6 \text{\ \ \ \ (by \,$k\leq 2^k$).}\rlap{\hbox to 280 pt{\hfill\qEd}}
\end{align*}

\noindent Now let us show the following:


\begin{clm}\label{intersec}
The function $h=\text{Inv}_g\dotminus 1$ is a true \,\C\, order, and 
for any $\xi\in\cantor$,

 \centerline{$\xi\in\bigcap_{n\in\n}G_n \ \Rightarrow\ B(\xi\restriction i)\geq 2^{h(i)}$ \ i.o.}
\end{clm}

\begin{proof} \ If $h=\text{Inv}_g\dotminus 1$, then for $n\geq 1, \ \text{Inv}_h(n)=\text{Inv}_{\text{Inv}_g}(n+1)$. Hence by Claim~\ref{inv-inv}(b), $h$ is a true \,\C\ order. Now for the second statement, we provide an argument for our precise definition of the martingale $B$, but the line of proof should be as in Schnorr's original demonstration.

Let $\xi\in\bigcap_{n\in\n}G_n$. For any $n\in\n$, there must exist \,$i_n\in\n$ \,such that \,$\xi\restriction i_n\in X_n$. The inclusion \,$[\xi\restriction i_n]\subseteq G_n$ \,implies \ $2^{-i_{n}}\leq \mu(G_n)\leq 2^{-2n}$. Hence for any $n\in\n, \ i_n\geq n$. We check 
\begin{align}
\text{for any }n>g(0), \ B(\xi\restriction i_n)\geq 2^{h(i_n)}.
\end{align}
We now assume $n>g(0)$, then also $i_n>g(0)$ and hence $\text{Inv}_g(i_n)\geq 1$. By definition of $\text{Inv}$, $g(\text{Inv}_g(i_n)-1)<i_n.$ Let $k_n=h(i_n)=\text{Inv}_g(i_n)-1$. Since $g(k_n)<i_n$, necessarily \,$\xi\restriction i_n\notin X^M_{n,g(k_n)}$. We know \,$\xi\restriction i_n\in X_n$, hence necessarily \,$\xi\restriction i_n\in X_n\setminus X^M_{n,g(k_n)}$. \,This gives \,$\mu(C^{k_n}_n|\xi\restriction i_n)=1$.
\\Therefore  \ \,$B(\xi\restriction i_n)=\sum_{k,n}2^k\mu(C^k_n|\xi\restriction i_n)\geq 2^{k_n}=2^{h(i_n)}$.
\end{proof}

\noindent Our  goal now is to find a function \,$F:\words\times\n\to\q_2$ \,which \,\expo(\C)\,-approximates \,$B$. Before developing the whole proof, we summarize the argument:
\begin{itemize}
\item for \,$x\in\words,\ \,B(x)$ \,is an infinite sum of real terms \,$2^k\mu(C^k_n|x)$. We first truncate \,$B(x)$ to obtain a finite sum \,$B_2(x,i)\,=\,\Sigma_{(n,k)\in A(|x|,i)}2^k\mu(C^k_n|x)$, 
for \,$A(|x|,i)\ \text{finite}\,\subseteq\,\n\times\n$ \,appropriately bounded so that
 \,$B_2(x,i)$ approximates \,$B(x)$\, within \,$2^{-(i-1)}$.
\item The second step consists in replacing, in the finite sum \,$B_2(x,i)\,$, each term\\ \,$2^k\mu(C^k_n|x)\,=\,2^k2^{|x|}\mu([X_n\setminus X^M_{n,g(k)}]\cap[x])$ \ by the term \ $2^k2^{|x|}\mu([X^M_{n,\bar g(|x|,i)}\setminus X^M_{n,g(k)}]\cap[x])$\, for an adequate function \,$\bar g$ \,in \,\C.\\
 By switching from measures of open sets \ $[X_n\setminus X^M_{n,g(k)}]\cap[x]$ \ to measures of clopen sets \ $[X^M_{n,\bar g(|x|,i)}\setminus X^M_{n,g(k)}]\cap[x]$, \,we shall obtain a sum \ $F(x,i)\,=\,\Sigma_{(n,k)\in A(|x|,i)}2^k\mu([X^M_{n,\bar g(|x|,i)}\setminus X^M_{n,g(k)}]\cap[x])$ \ in \,$\q_2$ \ with the expected approximation properties.
\item Moreover the bounds on \ $A(|x|,i)$ (polynomial in $(|x|,i)$) \,and the fact that \,$\bar g$\, belongs to \,\C\, will imply that \ $F:\words\times\n\to\q_2$ \ belongs to \,\expo(\C).
\end{itemize}

\noindent Hence let \ $B(x)\,=\,\sum_{k,n}2^k\mu(C^k_n|x)$.
\begin{enumerate}[label=(\textbf{\roman*})]
\item
  We first bind the integer $n$ in the sum.
  Having set \,$N(r,i,k)=r+i+2k+3$, we consider

\centerline{$B_1(x,i)\,=\,\sum_k\sum_{n\leq N(|x|,i,k)}2^k\mu(C^k_n|x).$}
Then \begin{align}\label{prem}
\begin{split}
0\,\leq\,B(x)-B_1(x,i)&=\,\sum_k\sum_{n}2^k\mu(C^k_n|x)-\sum_k\sum_{n\leq N(|x|,i,k)}2^k\mu(C^k_n|x),\\
&=\,\sum_k\sum_{n>N(|x|,i,k)}2^k\mu(C^k_n|x),\\
&\leq\,\sum_k 2^k2^{|x|}\sum_{n>N(|x|,i,k)}\mu(G_n),\\
&\leq\,2^{|x|}\sum_k 2^k2^{-N(|x|,i,k)} \text{ \espa (by \,$\mu(G_n)\leq 2^{-n}$)}\\
&\leq\,2^{-i-2}.
\end{split}
\end{align}
 \item \,Let us deal now with $k$.
 We set \ $K(r,i)=r+i+4$ \,and consider

 \centerline{$B_2(x,i)\,=\,\sum_{k\leq K(|x|,i)}\sum_{n\leq N(|x|,i,k)}2^k\mu(C^k_n|x).$}
\noindent Our goal is to show \ $0\,\leq\,B_1(x,i)-B_2(x,i)\,\leq\,2^{-i-2}$.
\\Let us set \ $\varepsilon_k(x,i)\,=\,\sum_{n\leq N(|x|,i,k)}2^k\mu(C^k_n|x)$.

\begin{clm}\label{epsi}
If \,$k>|x|+i+4$, \,then \,$\varepsilon_k(x,i)\,\leq 2^{|x|-k+2}$.
\end{clm}

\begin{proof} \hfill \begin{itemize}[label=$-$]
\item 
 If $n\leq 2k$, then by (\!~\ref{bound-below-2k}), \,$\mu(C^k_n\cap [x])\leq 2^{-3k}$.
\item If $n>2k$, then by (\!~\ref{bound-2n}), \,$\mu(C^k_n\cap [x])\leq 2^{-2n}<2^{-4k}$.
\end{itemize}
Hence for all $k,n\in\n,\ \mu(C^k_n\cap [x])\leq 2^{-3k}$.

The hypothesis \,$k\,>\,|x|+i+4$\, implies \,$N(|x|,i,k)\,=\,|x|+i+2k+3\,\leq\, 3k$.
 Hence we deduce that for any $k\in\n$, \ 
$\varepsilon_k(x,i)\,\leq\,N(|x|,i,k)\,2^k2^{-3k}2^{|x|}\,\leq \, 3k2^{-2k+|x|}\,\leq\,2^{2+|x|-k}$.
\end{proof}

\noindent Therefore \begin{align}\label{deux}
0\,\leq\, B_1(x,i)-B_2(x,i)\,=\,\sum_{k>|x|+i+4}
\varepsilon_k(x,i)\ \leq\,
\sum_{k>|x|+i+4}
2^{2+|x|-k}
\,\leq\, 2^{-i-2}.
\end{align}
\item We now define the function $\bar{g}:\n\times\n\to\n$ used to switch from the open sets \,$C^k_n$ \,to clopen sets.

\begin{defi}\label{defiF}
Let \,$\bar g(r,i)=f(9r+9i+32)\, $ and  \,$D_n^{k,i,r}=[X_{n,\bar{g}(r,i)}\setminus X_{n,g(k)}] $. We set

\espa $F(x,i)\,=\,\sum_{k\leq K(|x|,i)}\sum_{n\leq N(|x|,i,k)}2^k\mu(D^{k,i,|x|}_n|x)$.
\end{defi}

Our goal is to obtain \
$0\,\leq\,B_2(x,i)-F(x,i)\,\leq \,2^{-i-1}$.
\begin{itemize}
\item We first check that for $k\leq K(r,i)$, \ $\bar{g}(r,i)\geq g(k)$:

$\bar{g}(r,i)\,=\,f(9r+9i+32)\,\geq \,f(5(r+i+4))\,=\,f(5K(r,i))\,\geq\,f(5k)\,=\,g(k)$.
\item Let us check that for \,$k\leq K(r,i)$\, and \,$n\leq N(r,i,k)$, \ $\bar{g}(r,i)\geq f(n+(2r+2i+4k+6))$.
 One computes \,$N(r,i,K(r,i))=3r+3i+11$. Hence 

  $\left(\begin{array}{l}k\leq K(r,i)\\
 n\leq N(r,i,k)
 \end{array}\right. \Rightarrow \left(\begin{array}{l}
 k\leq r+i+4\\
 n\leq 3r+3i+11
 \end{array}\right.
 \Rightarrow \ n+(2r+2i+4k+5)\,\leq\, 9r+9i+32$.

 One deduces \begin{align}\label{delta-f}
 f(n+(2r+2i+4k+5))\leq f(9r+9i+32)\leq \bar{g}(r,i).
 \end{align}
\item For $k\leq K(|x|,i)$ and \,$n\leq N(|x|,i,k)$, let us set \ $\Delta(n,k,x,i)=\mu(C^k_n\cap [x]) -\mu(D^{k,i,|x|}_n\cap[x])$. One obtains:
\begin{align}\label{delta}
\begin{split}
\Delta(n,k,x,i)&=\,\mu((C^k_n\setminus D^{k,i,|x|}_n)\cap[x])\\
&\leq\,\mu(C^k_n\setminus D_n^{k,i,|x|})\\
&\leq\,\mu([(X_n\setminus X^M_{n,g(k)})\setminus(X^M_{n,\bar{g}(|x|,i)}\setminus X^M_{n,g(k)})])\\
&\leq\,\mu([X_n\setminus X^M_{n,\bar{g}(|x|,i)}])\\
&\leq\,\mu([X_n\setminus X^M_{n,f(n+(2|x|+2i+4k+5))}]) \text{\esp(by (\!~\ref{delta-f}))}\\
&\leq\,2^{-(2|x|+2i+4k+5)}
\end{split}
\end{align}
\end{itemize}
We derive:

\begin{align}\label{trois}
\begin{split}
0\leq B_2(x,i)-F(x,i)
&=\,\sum_{k\leq K(|x|,i)}\sum_{n\leq N(|x|,i,k)}2^k2^{|x|}\mu((C^k_n\setminus D^{k,i,|x|}_n)\cap[x])\\
&=\,\sum_{k\leq K(|x|,i)}\sum_{n\leq N(|x|,i,k)}2^{k+|x|}\Delta(n,k,x,i)\\
&\leq \  \sum_{k\leq K(|x|,i)}\sum_{n\leq N(|x|,i,k)}2^{k+|x|}2^{-(2|x|+2i+4k+5)} \text{\espa (by\ (\!~\ref{delta}))}\\
&\leq \  \sum_{k\leq K(|x|,i)}N(|x|,i,k)2^{-(|x|+2i+3k+5)} \\
&\leq \ \ \sum_{k}2^{|x|+i+2k+3}2^{-(|x|+2i+3k+5)} \\
&\leq \ \ 2^{-i-1}.
\end{split}
\end{align}
\end{enumerate}
Combining (\!~\ref{prem}),\,(\!~\ref{deux}) and (\!~\ref{trois}), we deduce:

\begin{clm}\label{B-approx}
For any $x\in\words$ and $i\in\n$, \ $0\,\leq\,B(x)-F(x,i)\,\leq\, 2^{-i}$.
\end{clm}

\noindent It now remains to evaluate the complexity of $F$. Clearly if $f$ is in \,\C, \,then $\bar{g}$ is also in \,\C\,.
Given $k,n,i,x$, to compute \,$\mu(D_n^{k,i,|x|}\cap [x])$, one has to check for each finite sequence $z$ of length $\leq\bar{g}(|x|,i)$ compatible with $x$ whether it belongs to \,$X^M_{n,\bar{g}(|x|,i)}\setminus X^M_{n,g(k)}$ and to compare it with $x$.
\begin{itemize}[label=$-$]
\item
\ If there is  $z\in X^M_{n,\bar{g}(|x|,i)}\setminus X^M_{n,g(k)}$ \,such that \,$z\curl x$, then $\mu(D_n^{k,i,|x|}\cap [x])$ \,is \,$2^{-|x|}$.
\item  \ Otherwise one adds all $2^{-|z|},$ \,for \,$x\cur z$ \,with \,$z$ \,in\, $X^M_{n,\bar{g}(|x|,i)}\setminus X^M_{n,g(k)}$ \,to obtain the measure. All intermediate (and the final) sums can be coded, for some constant $d$, by strings of length \,$\leq d(\bar{g}(|x|,i)$ (the total measure is $\leq 1$).
\end{itemize}

Hence the function \ $\varphi:\begin{array}[t]{lcl}\n^{3}\times\words&\!\!\to\!\! &\q_2\vspace{-1mm}\\
(k,n,i,x)&\!\!\mapsto \!\!&\mu(D_n^{k,i,|x|}\cap [x])
\end{array}$ \,is in \,\expo(\C).

Since for some constant $c$, our bounds $K(|x|,i)$ and $N(|x|,i,k)$ are \,$\leq c(|x|+i)$, we deduce 

\begin{clm}\label{F-expc}
$F$ belongs to the class  \,\expo(\C).
\end{clm}

\noindent We can now conclude the proof of (a) by Claims~\ref{Bmart},
~\ref{intersec}, ~\ref{B-approx} and ~\ref{F-expc}.\\

\item Let now  \,\C\, be the class \,\pspace .\\
We assume the sequence \,$([X_n])_{n\in\n}$\, associated with a machine $M$, satisfies for some (strictly increasing) polynomial function $f$:
 \begin{itemize}[label=$-$]
\item $\mu([X_n])\leq 2^{-2n},$
\item $\mu([X_n])-\mu([X_{n,f(n+i)}^{M,\text{space}}])\leq 2^{-i}.$
\end{itemize}
%
One defines $g$, $h$ and $\bar g$ from $f$  as in (a). They are all polynomial functions. $C^k_n$ is now the open set \,$[X_n\setminus X_{n,g(k)}^{M,\text{space}}]$ \,and one also considers the martingale \ $B(x)=\sum_{n,k} 2^k\mu(C^k_n\,|x)$.
\\(Since $|x|>s$ implies \,$x\notin X_{n,s}^{M,\text{space}}$) one obtains the equivalent of Claim~\ref{intersec}: $h$ is a true \,\pspace\, order and for for any $\xi\in\cantor$,

 \centerline{$\xi\in\bigcap_{n\in\n}G_n \ \Rightarrow\ B(\xi\restriction i)\geq 2^{h(i)}$ \ i.o.}
 \noindent We set \,$D_n^{k,i,r}\!=[X_{n,\bar{g}(r,i)}^{M,\text{space}}\setminus X_{n,g(k)}^{M,\text{space}}]$ and define the approximating function $F$ as in (a) with the $D_n^{k,i,r}$'s.

To compute \,$\mu(D_n^{k,i,|x|}\cap [x])$, \,we also enumerate all sequences $z$ of length $\leq \bar{g}(|x|,i)$ according to $\curllex$, and (using  counters) we test whether $z\in D_n^{k,i,|x|}$.  But this time we only keep track of the last browsed sequence $z$ and of the partial measure \ $\mu([(X_{n,\bar{g}(|x|,i)}\setminus X_{n,g(k)})\cap\{t: t\curllex s\}]\cap[x])$.
As in (a), we know this partial measure is coded by a string of length $\+O (\bar{g}(|x|,i)))$.
\\Hence  \ $\varphi:\begin{array}[t]{lcl}\n^{3}\times\words&\!\!\to\!\! &\q_2\vspace{-1mm}\\
(k,n,i,x)&\!\!\mapsto \!\!&\mu(D_n^{k,i,|x|}\cap [x])
\end{array}$ \,is in \,\pspace.
\\ We deduce that $F$ is in \,\pspace\, and conclude the proof of (b) as above.\qedhere\end{enumerate}\end{proof}


\noindent By Proposition~\ref{ml-mart} and lemma~\ref{q2appro} about approximation, we derive:

\begin{prop}\label{ML-mart-q2}\hfill
\begin{enumerate}[label=\({\alph*}]
\item Let \,\C\ \,be one of our time-complexity classes. If \,$(G_n)_{n\in\n}$\, is an ML-\,\C\,-S test, then there
 exist a martingale $d:\words\to\q _2$ \,in \,\expo(\C)\, and
  a true \,\C\, order $h$ such that

\centerline{for any $\,\xi\in\cantor $, \  $\xi\in\bigcap_{n\in\n}G_n\ \Rightarrow\ B(\xi\restriction i)\geq 2^{h(i)}$ \ \ i.o.}
\item  Given an ML-\,\pspace\,-S test \,$(G_n)_{n\in\n}$,\, there exist a  martingale $d:\words\to\q _2$ in \,\pspace\, and a true \,\pspace\, order $h$ satisfying the above implication for any $\xi\in\cantor$.
\end{enumerate}
\end{prop} 

\noindent The opposite direction - from martingales to Martin-L\"of tests - is easier and can be obtained through a simple adaptation of existing arguments.

\begin{prop}\label{mart-to-ML}\hfill
\begin{enumerate}[label=\({\alph*}]
\item  Let \,\C\ \,be one of our time-complexity classes. From a martingale \,$d:\words\to\q_2$ \,in\, \C\, and a true \,\C\, order $g$, one can define an ML-\,\C\,-S test \,$(G_n)_{n\in\n}$ \,such that 

\centerline{for any $\,\xi\in\cantor $, \  $ d(\xi\restriction i)\geq 2^{g(i)}\text{ \ i.o.}\ \ \Rightarrow\ \  \xi\in\bigcap_{n\in\n}G_n$.}
 \item Given a \,\pspace\, martingale $d$ and a true \,\pspace\, order $g$, one can construct an ML-\,\pspace\,-S test \,$(G_n)_{n\in\n}$,\, satisfying for any $\xi\in\cantor$, \,the above implication.
 \end{enumerate}
\end{prop}

\begin{proof} Let $d$ and $g$ be respectively the martingale and the order. We can assume $d(\emptyset)\leq 1$.

Let us consider the set \ $X=\big\{(n,x)\in\n\times\words:\,d(x)\geq 2^{g(|x|)}\geq 2^n\big\}$, \,and for $n\in\n$, the associated set \ $X_n=\big\{x\in\words : (n,x)\in X\big\}$.

Setting for $n\in\n, \ G_n=[X_n]$, one obtains by classical arguments (~\cite[7.1.7]{dohi},~\cite[7.3.3]{nies}) that \,$(G_n)_{n\in\n}$ \,is an ML test and that for any $\xi\in\cantor$,

\centerline{$ d(\xi\restriction i)\geq 2^{g(i)}\text{ \ i.o.}\ \ \Rightarrow\ \  \xi\in\bigcap_{n\in\n}G_n$.}

We must now check that \,$(G_n)_{n\in\n}$\, is an ML-\,\C\,-S test. We shall make explicit the use of $\text{Inv}$ in Schnorr's original proof and add a few lines to define the controlling function.

\begin{enumerate}[label=\({\alph*}]
\item Let \,\C\ \,be one of our time-complexity classes. To deal with prefix-free sets, we consider minimal strings for $\curl$:
\begin{align*}
\ \llap{\text{Let\ }}Y&=\ \big\{(n,x)\in X :\forall y \,\cur x\  (n,y)\notin X\big\}\\
&=\ \big\{ (n,x)\in\n\times\words : d(x)\geq 2^{g(|x|)}\geq 2^n\ \land\ \forall y\,\cur x\ \neg (d(y)\geq 2^{g(|y|)}\geq 2^n)\big\}.\qquad
\end{align*}
For some constant $d$ and some function $f'\in F_C$, membership of $(n,x)$ in $Y$ can be checked in time $\leq f'(|x|+n)$. Hence one can define a machine  $M$ and a strictly increasing  function $f\in$\ \C\, such that for any $(n,x)\in\n\times\words$,
\begin{align*}
(n,x)\in Y &\ \Leftrightarrow \ M(n,x)\downarrow\\
&\ \Leftrightarrow\ (n,x)\in \text{dom}(M_{f(|x|+n)}).
\end{align*}
Setting, for $n\in\n$, \ $Y_n=\{x\in\words :(n,x)\in Y\}$. Then $[X_n]=[Y_n]$ and  we get, for $n\in\n,$\linebreak$ \,x\in\words ,$\\\centerline{$  x\in Y_n\ \Leftrightarrow\ x\in Y^M_{n,f(|x|+n)}$.}

For all $n,k\in\n$, one has:
\begin{align}\label{primo}
Y_n\cap\{x\in\words :|x|\leq \text{Inv}_g(k)\}\ \subseteq\ Y_{n,f(\text{Inv}_g(k)+n)}.
\end{align}

For any $m\in\n,\ g(\text{Inv}_g(m))\geq m$, hence
\begin{align}\label{deusio}
\begin{split}
Y_n\cap\big\{x\in\words :|x|>\text{Inv}_g(k)\big\}&\subseteq\ \{x\in\words : d(x)\geq 2^{g(|x|)}\geq 2^{g(\text{Inv}_g(k))}\}\\
&\subseteq\ \big\{x\in\words : d(x)\geq 2^k\big\}.
\end{split}
\end{align}

Now by~\cite[6.3.3]{dohi},~\cite[7.1.9]{nies}, since \  $d(\emptyset)\leq 1$,
\begin{align}\label{tertio}
\mu 
( [\{x\in\words
 : d(x)\geq 2^k \}])
\leq 2^{-k}.
\end{align}

Let us set $\bar{g}(r)=f(\text{Inv}_g(r)+r)$, for $r\in\n$. Since $g$ is a true \,\C\, order, $\bar{g}$ is in \,\C.
 For any $k,n\in\n,\ \bar{g}(k+n)\geq f(\text{Inv}_g(k)+n)$, we  thus deduce:
\begin{align*}
\mu([Y_n])-\mu([Y_{n,\bar{g}(k+n)}])&=\ \mu([Y_n\setminus Y_{n,\bar{g}(k+n)}]\\
&\leq\ \mu([Y_n\cap\big\{x\in\words :|x|>\text{Inv}_g(k)\big\}])\text{\esp (by (\!~\ref{primo}))}\\
&\leq\ 2^{-k}\esp\text{ (by (\!~\ref{deusio}) and
(\!~\ref{tertio}))}.
\end{align*}
Therefore \,$(G_n)_{n\in\n}\,=\,([X_n])_{n\in\n}\,=\,([Y_n])_{n\in\n}$ \,is an ML-\,\C\,-S test.\\

\item Let now \,\C\, be \,\pspace\, and let $d$  in \,\pspace\ \,and $g$ a true  \,\pspace\, order satisfy the implication. The set $\,Y\,$ being defined as above, membership  in $Y$ can be tested in polynomial space (either by using a time-counter in binary or by Claim~\ref{bound-inv}). Hence there exist a machine $M$ and a polynomial function \,$f$\, so that for any $(n,x)\in\n\times\words$,
\begin{align*}
(n,x)\in Y &\ \Leftrightarrow \ M(n,x)\downarrow\\
&\ \Leftrightarrow\ (n,x)\in \text{dom}(M_{f(|x|+n)}^{\text{space}}).
\end{align*}
One then concludes as above (with $M_s^{\text{space}}$ and $Y^{M,\text{space}}_{n,s}$\, instead of $M_s$ and $Y^M_{n,s}$).\qedhere
\end{enumerate}\end{proof}

\begin{rem}\label{2^g} 
Replacing the classical requirement \,``$d(\xi\restriction_i)\geq g(i)\ \,i.o$" \,by \,``$d(\xi\restriction_i)\geq 2^{g(i)}\ \,i.o$" \,allowed us to consider the function \,$\bar g(r)=f(Inv_g(r)+r)$ \,instead of the function \,$f(Inv_g(2^r)+r)$. This was essential to get \,Proposition~\ref{mart-to-ML} \ for \ \C\,=\,\pol, \expo\ \,or \,\pspace.
\end{rem}

\noindent Combining Propositions~\ref{ML-mart-q2} and~\ref{mart-to-ML}, we deduce:

\begin{thm}\label{thm-ml-mart1}\hfill
\begin{enumerate}[label=\({\alph*}]
\item Let \,\C\, be one of our time-complexity classes. \ Then for any \,$\xi\in\cantor$,
\\$\xi$\,is\,martingale-\,\expo(\C)-S\,random\,$\Rightarrow\xi$\,is\,ML-\,\C-S random\,$\Rightarrow\xi$\,is\,martingale-\,\C\,-S random.
\item Let \,\C\, be the class \,\pspace, \toexp \ \,or \primrec. \, Then for any sequence $\xi\in\cantor$,

\centerline{$\xi$ is  $\text{martingale-\,\C\,-S random}\ \Leftrightarrow\ \xi$ is $ \text{ML-\,\C\,-S random}.$}
\end{enumerate}
\end{thm}

\noindent Finally merging theorems~\ref{thm1} and~\ref{thm-ml-mart1}, we obtain:

\begin{thm}
Let \,\C\, be the class \,\pspace,\ \toexp\ \,or \,\primrec. \ Then for any $\xi\in\cantor$,\\
$ \xi$ is$\text{ ML-\,\C\,-S random} \Leftrightarrow \xi$ is$ \text{ Kolmogorov -\,\C\,-S random}\Leftrightarrow \xi$ is$\text{ martingale-\,\C\,-S random}$.
\end{thm}

\section{Separation.}

To justify our previous work, we now differentiate Schnorr randomness from (martingale)-\,\primrec -S randomness by appealing to the following notion:

\begin{defi}\label{lutz-random}
Let \,\C\, be a class of functions. An infinite binary sequence $\xi$ is \,\C\ random if any martingale $d:\words\to\q_2$ in  \,\C\, fails on $\xi$ (i.e. the set $\big\{d(\xi\restriction i):i\in\n\big\}$ \,is bounded in $\n$).
\end{defi}

If \,\C\, is the class of computable functions, then the above notion is ``computable randomness" (Schnorr).
 When \,\C\, is \,\pol, this is ``p-randomness" (Lutz).
 
The following argument was suggested by one of the (anonymous) referees:
let \,$A:\n\to\n$\, be a computable function dominating all primitive recursive functions. Then by classical results (see ~\cite[3.9.7]{ambomayo} for a precise statement), there exists a computable sequence \,$\xi\in\cantor$\, which is FDTIME($A(n)$) random. \ 
$\xi$\, is thus \,\primrec\,-random and hence (martingale)-\,\primrec -S random. Therefore Schnorr's randomness is strictly stronger than (martingale)\,\primrec -S randomness.

Our original argument was based on the notion of \,ML-\primrec -S randomness. The method - though laborious - could be extended to prove the assertions:
\begin{itemize}[label=$-$]
\item
 ML-\primrec -S randomness\,$>$\,ML-\toexp\,-S randomness\, \ and 
 \item  ML-\toexp\,-S randomness\,$>$\,ML-\expo -S randomness.
 \end{itemize}
But we could not deduce that \,ML-\,\expo\,-S randomness is strictly stronger than \,ML-\,\pol\,-S randomness whereas this can be done for the martingale corresponding notion of S randomness. The martingale approach seems better suited to low time-complexity classes, we shall thus build on the important amount of work developed around the notion of martingale in the field of Resource Bounded Randomness.

In this section, we shall restrict ourselves to time-complexity classes and we shall focus on the martingale definition of S-randomness. For such a  class \,\C,  the expression ``martingale-\,\C\,-S random" will be abbreviated to ``\,\C\,-S random".  As \,\C\, rises among our time-complexity classes, the notion of \,\C\,-S randomness gets strictly stronger. We shall  compare our notion of \,\pol\,-S randomness with Lutz notion of p-randomness, Wang's notion of (\pol,\pol)-S randomness and and we shall also contrast the notion of \,\primrec\,-S randomness with the notion of BP-randomness developed by Buss, Cenzer and Remmel.

\subsection{\,\C\,-S randomness\ \ and \,\C\ randomness.}

%

Let us  mention first the work of Wang\\ \cite{wang1999,wang2000} who studied a  version of Schnorr randomness for the class \,\pol\ (termed (\pol,\pol)-S randomness) and proved it to be weaker than the notion of p-randomness~\cite[Thm 8]{wang2000}.

\begin{defi}[~\cite{wang2000}]
Let \,\C\ \,be a class of functions. An infinite sequence $\xi$ is (\C,\C)-S random iff for any martingale $F$ and any order $h$ both in \,\C, \ $F(\xi\restriction i)<h(i)$ \,a.e.
\end{defi}

 His notion is stronger than ours because he allows all orders in \,\C, not restricting to true \,\C\, orders. (If one is not concerned with the status of  the inverse of the order, our condition ``$d(\xi\restriction i)< 2^{h(i)}$ a.e." and the classical one ``$d(\xi\restriction i)< h(i)$ a.e." yield the same notion of randomness  for our classes \,\C).
 A consequence of his definition is that for computable infinite sequences, p-randomness and (\pol,\pol)-S randomness coincide~\cite[Cor. 17]{wang2000}.

 Building on his results and techniques, we shall show that our definition allows more variety inside the set of computable infinite sequences.
 
Let \,\C\,  be one of our time-complexity classes. To separate \C\ randomness from \,\C\,-S randomness inside the set of computable sequences (and to obtain the tableau of the introduction) we shall rely on part (i) of the following proposition; the remaining cases (ii)-(iv) add  precision, showing that the sequence which is \,\C -S random but not \,\C\ random can be taken ``right above \C ". 

%

%
For $\xi\in\cantor$ and $g\in\n^\n$  time-constructible, we say that $\,\xi$ \,belongs to \,\emph{FDTIME}$(g(n))$ \,if the function \,$n\mapsto\xi(n)$\, belongs to \,\emph{FDTIME}$(g(n))$ \,($n$ under unary representation).\\

\begin{prop}~\label{csrandom-not-crandom}\hfill
\begin{enumerate}[label=\({\roman*}]
\item There is a computable \,$\xi\in\cantor$\, which is \,\primrec\,-S random but not \,\pol\ random.
\item There exists \,$\xi\in\cantor$ \,in \,FDTIME$(n^{\lceil\log\, n\rceil})$ \,which is \,\pol\,-S random but not \,\pol\ random.
\item There is \,$\xi\in\cantor$ \,in \,FDTIME$(2^{n^{\lceil\log\, n\rceil}})$ \,which is \,\expo\,-S random but not \,\pol\ random.
\item There is \,$\xi\in\cantor$ \,in \,FDTIME$(T(\lceil\log\, n\rceil,n))$ \,which is \,\toexp\,-S random but not \,\pol\ random.
\end{enumerate}
\end{prop}
 
\noindent  Our proof will be based on enumerations of martingales and of true \,\primrec\, orders. We thus propose without proof a few definitions and classical (or easy) facts:\\ 
 \begin{defi}~\label{easy}
 Let \,\C\, be one of our time-complexity classes. We assume \,$G_C:\n\times\words\to\q_2$\,  enumerates all functions $f:\words\to\q_2$ in \,\C\,, and $g_c:\n\to\n$ \,  strictly increasing  is such that \,$G_C\in\text{FDTIME}(g_C(n))$.
\begin{enumerate}[label=\({\alph*}]
\item
  Then there is \,$d_C:\n\times\words\to\q_2$ \,enumerating all martingales $d$ in \,\C\ with $d(\emptyset)\leq 1$  which is such that \,$d_C\in\text{FDTIME}(ng_C(n))$ (see for example~\cite{amboter}).
\item Let us define the martingale \ $\Phi_C(x)=\sum_{e\in\n}2^{-e}d_C((e)_0,x)$. It can be approximated by \,$f_C:\words\times\n\to\q_2$ \,defined by \,$f_C(x,i)=\sum_{e\leq i+|x|}2^{-e}d_C((e)_0,x)$ \,which belongs to \,$\text{FDTIME}(n^2g_C(2n))$.
\item By Lemma~\ref{q2appro}, there is a $\q_2$-valued martingale \,$\delta_C$ \,in \,$\text{FDTIME}(n^3g_C(5n))$ \,such that for any $x\in\words$, \ $\Phi_C(x)\leq \delta_C(x)\leq \Phi_C(x)+2$.
\end{enumerate}
\end{defi} 

\noindent The notation $(e)_0$ or $(e)_1$ refers to the inverses of the polynomial time bijection from $\n\times\n$ onto $\n$.
If \,\C\, is \,\primrec, \,then both $G_C$ and $g_C$ can be taken recursive. Hence in (c), we only assert that $\delta_{\text{\primrec}}$ \,is recursive. In all cases, $g_C$ will be time-constructible. Here are some possible choices for ($G_C$ and) $g_C$:\\

\begin{fact}~\label{fact1} 
One can take\hfill
\begin{enumerate}[label=(\roman*)]
\item $g_{\text{\primrec}}$ \,recursive
\item $g_{\text{\pol}}(n)=n^{\lceil(2/3)log \,n\rceil}$
\item $g_{\text{\expo}}(n)=2^{\lceil(1/2)log\, n\rceil}$
\item $g_{\text{\toexp}}(n)= T(\lceil(1/2)log\, n\rceil,n)$.
\end{enumerate}
\end{fact}

\begin{proof} We give a few details for (ii), the other cases are very similar.

\noindent We want to define \,$G_{\text{\pol}}:\n\times\words\to\q_2$\, enumerating all polynomial time functions from $\words$ into \,$\q_2$.
Let $M$ be a universal machine which  for some constant $c\in\n$,  simulates the computation of $t(n)$ steps of the machine $M_e$ (with program $e\in\n$) in \ $ce\,t(n)\lceil\log\, t(n)\rceil$ steps.
 We define an algorithm for $G_{\text{\pol}}$:
On input $(e,x)$, one computes $(e)_0,\ (e)_1$ and $\lfloor(1/2)\log (e)_1\rfloor$. Then $M$ simulates $M_{(e)_0}$ on $x$ during \,$|x|^{\lfloor(1/2)\log (e)_1\rfloor}$ steps. If $M_{(e)_0}$ halts, then $G_{\text{\pol}}$ outputs the result of the computation, otherwise it outputs 0.\\
One can check that \,$G_{\text{\pol}}$ \,enumerates all polynomial time functions and that the function \,$G_{\text{\pol}}$\, belongs to \,$\text{FDTIME}(n^{\lceil(2/3)\log n\rceil})$ (the input $(e,x)$ has size $n=e+|x|$).
\end{proof}

Refering to the martingale $\delta_{\text{\C}}$ in Definition~\ref{easy}, one derives from the previous values:\\ 

\begin{fact}~\label{fact2}\hfill
\begin{enumerate}[label=(\roman*)]
\item
 \ $\delta_{\text{\primrec}}$ \ is recursive
\item $\delta_{\text{\pol}}\in\text{FDTIME}(n^{\lceil(3/4)\log\, n\rceil})$
\item $\delta_{\text{\expo}}\in\text{FDTIME}(2^{n^{\lceil(3/4)\log\, n\rceil}})$
\item $\delta_{\text{\toexp}}\in\text{FDTIME}(T(\lceil(3/4)\log\, n\rceil,n)$
\end{enumerate}
\end{fact}

Let \ $(\varphi_e)_{e\in\n}$ \ be the usual effective enumeration of partial recursive functions from (a subset of) $\n$ into $\n$. We obtain the following:\\
\begin{lem}~\label{prim-inv}
There exists a (total) recursive function $l\in\n^\n$ such that \ $(\varphi_{l(e)})_{e\in\n}$ \ is an enumeration of all inverses of strictly increasing primitive recursive functions in \,$\n^\n$.
\end{lem}

\begin{proof} Let \ $(\varphi_{f(e)})_{e\in\n}$ \ be an enumeration of all primitive recursive functions, with $f$ recursive.\\
We define the partial recursive function $\lambda$ by \ 
$\begin{cases}
\lambda(e,0)=\varphi_e(0)\\
\lambda(e,n+1)=\max (\varphi_e(n+1),\lambda(e,n)+1).
\end{cases}$\\
(if one of the two arguments of $\max$ is undefined, $\max$ is undefined)
\\ There is $g$ recursive such that for any $e,n\in\n ,\ \lambda(e,n)=\varphi_{g(e)}(n)$.
\\Hence \ $(\varphi_{g(f(e))})_{e\in\n}$ \ is an enumeration of all strictly increasing primitive recursive functions.
\\ We now define the partial recursive function \,$\lambda '$\, by \ $\lambda '(e,n)=\text{least }k\ \varphi_e(k)\geq n$.
\\Then again there is $g'$ recursive such that for any $e,n\in\n ,\ \lambda'(e,n)=\varphi_{g'(e)}(n)$.
Setting $\,l(e)=g'(g(f(e))$, we obtain that \,$(\varphi_{l(e)})_{e\in\n}$ \,is an enumeration of inverses of strictly increasing primitive recursive functions.
\end{proof}

\begin{fact}~\label{enum-mart-order}\ 
For \,\C\, one of our time-complexity classes and $e\in\n$, let $d_{C,e}$  be the martingale defined by \,$d_{C,e}(x)=d_C((e)_0, x)$, for $x\in\words$ ($d_C$ was defined in~\ref{easy}), and let $h_e$ be the order $\varphi_{l((e)_1)}$.\\
Then \ $(d_{C,e},h_e)_{e\in\n}$ \ is an enumeration of all couples \,$(d,h)$\, where $d$ is a martingale in \,\C\, such that $\,d(\emptyset)\leq 1$ \,and $h$ is the inverse of a strictly increasing primitive recursive function.
\end{fact}

\begin{proof} (Of Proposition~\ref{csrandom-not-crandom})\ \ \ We adapt and simplify Wang's arguments (see~\cite[Thm 5]{wang1999}): there is no need to encode non-recursive information into the sequence $\xi$ which separates the notions of \,\C\,-S randomness and of \,\C\ randomness (being an order is a non effective notion whereas - as we saw - being the inverse of a strictly increasing primitive recursive function is an effective one).

We shall thus define by induction two functions $F$ and $T$ both in\, \pol, with \,$F:\words\to\q_2$\,  a martingale and \,$T:\words\to\n$\,  monotone (i.e. if $x\curl y$, \,then \,$T(x)\leq T(y)$).
\\Let \,$L$\, be a machine which computes the (total) function $l$ of Lemma~\ref{prim-inv}, and let \,$M_u$\, be a universal machine.

\esp - \textbf{Level 0}. \ Let $\begin{array}[t]{l}
F(\emptyset)=1\\
T(\emptyset)=0.
\end{array}$

\esp - \textbf{Level $\pmb{s+1}$}. \,We assume \,$F(x),\ T(x)$ \,are defined for $|x|\leq s$ \,and \,$T(x)\leq |x|$. As in~\cite{wang1999}, one distinguishes two cases:

\underline{Case 1:} \ For each $\,e\leq T(x), \ L$ \,on input \,$(e)_1$\, stops in $\,\leq |x|+1$ \,steps and there is \,$m_e\leq |x|$ \,such that $M_u$ on input \,$(L((e)_1),m_e)$\, stops in\, $\leq|x|+1$\, steps (outputting \,$\varphi_{l((e)_1)}(m_e)=h_e(m_e)$),  and such that\, $e+T(x)+3<h_e(m_e)$.
\\Then one sets \ $\left (\!\begin{array}{l}
F(x0)=2F(x)\\
F(x1)=0
\end{array}\right.$ \ \,and \  $\left(\!\begin{array}{l}
T(x0)=T(x)+1\\
T(x1)=T(x).
\end{array}\right.$

\underline{Case 2:} \ Otherwise one sets \ $F(x0)=F(x1)=F(x)$ \,and \,$T(x0)=T(x1)=T(x)$.
\\Let us note that \,$T(x)$\, is simply the number of times case 1 has occured along $x$. We also notice that $F$ and $T$ are computable in polynomial time.
\\The inductive definition of the infinite sequence $\xi_C$  is as follows:

\begin{defi}\label{def-xic}
Let $s\in\n$. We assume $\xi_C\restriction s$\, is defined.
\begin{enumerate}[label=\({\alph*}]
\item
 \ If \,$\xi_C\restriction s$\, is in case 1, then one sets \,$\xi_C(s)=0$.
\item Otherwise, one sets \,$\xi_C(s)=i$\, where $i\in\{0,1\}$ \,is such that\\ 
\centerline{$\delta_C((\xi_C\restriction s)i)\leq\delta_C((\xi_C\restriction s)(1-i))$.}
\end{enumerate}
\end{defi}

\noindent From this definition and Fact~\ref{fact2}, one deduces:\\

\begin{fact}~\label{complex-xic}\hfill
\begin{itemize}
\item If \,\C\ =\ \pol, \,then \,$\xi_C\in\text{FDTIME}(n^{\lceil\log\, n\rceil})$.
\item If \,\C\ =\ \expo, \,then \,$\xi_C\in\text{FDTIME}(2^{n^{\lceil\log\, n\rceil}})$.
\item If \,\C\ =\ \toexp, \,then \,$\xi_C\in\text{FDTIME}(T(\lceil\log\, n\rceil,n))$.
\item If \,\C\, is \,\primrec, \,then \,$\xi_C$\, is recursive.
\end{itemize}
\end{fact}


The great lines of the proof are essentially Wang's ones. A trustful reader can skip our proof. However since we simplified the argument (for instance, deleting mention of $F(x)$ in the definition of case 1) and added the machine $L$, we provide some arguments.

\begin{clm}\label{claim1}
Let $\alpha\in\cantor$. Then\,$ \alpha\restriction s$ is in case 1 infinitely often.
\end{clm}

\begin{proof}\ \ Let $\alpha\in\cantor$ and $s_0\in\n$ be fixed.\\
 Since $L$ defines a total function, there must exist $s_1\geq s_0$ such that for all $e\leq T(\alpha\restriction s_0),\ \,L$ \,on input $(e)_1$ stops in \,$\leq s_1+1$\, steps.\\
 Also since for $e\in\n$, $h_e$\ ($=\varphi_{L((e)_1)}$) \,is an order, there must exist $s_2\geq s_1$ such that for every $e\leq T(\alpha\restriction s_0)$, there is  $m_e\leq s_2$ such that \,$h_e(m_e)>e+T(\alpha\restriction s_0)+3$.
\\Hence there is $s\geq s_0$ such that $P(s)$ holds where\\
 $P(s)=\left\{\begin{array}{l}
\text{for each } e\leq T(\alpha\restriction s_0),\ L \text{ on input } (e)_1 \text{ stops} \text{ in } \leq s+1 \text{ steps,}\\
\text{there is } m_e\leq s \text{ such that } M_u
\text{ on input } (L((e)_1),m_e)\ \text{stops in } \leq s+1\\ \text{ steps outputting }
 o_e \ (=h_e(m_e)) \text{ which satisfies \ } e+T(\alpha\restriction s_0)+3<o_e.
\end{array}\right.$

\noindent Let $s_3=\min \{s\geq s_0:P(s)\}$. Then by construction \,$T(\alpha\restriction s_3)=T(\alpha\restriction s_0)$. Hence \,$\alpha\restriction s_3$\, is in case 1.
\end{proof}

\begin{clm}\label{claim2}
$\lim\limits_{s\to\infty}\, F(\xi_C\restriction s)\,=\,\lim\limits_{s\to\infty} T(\xi_C\restriction s)\,=\,+\infty$. Hence \,$\xi_C$\, is not \,\pol\ random.
\end{clm}

\begin{proof}\ \ By definition of $F,\,T\ \,$and\, $\xi_C$, for $s\in\n$,
\begin{itemize}[label=$-$]
\item $\xi_C\restriction s$ is in case 1, \,$F(\xi_C\restriction s+1)=2F(\xi_C\restriction s)$ \ and \ $T(\xi_C\restriction s+1)=T(\xi_C\restriction s)+1$,
\item when $\xi_C\restriction s$ is in case 2, \,$F(\xi_C\restriction s+1)=F(\xi_C\restriction s)$ \ and \ $T(\xi_C\restriction s+1)=T(\xi_C\restriction s)$.
\end{itemize}
Hence we can conclude by the previous claim.
\end{proof}

\begin{clm}\label{claim3}
For any $s\in\n,\ \,\delta_C(\xi_C\restriction s)<2^{T(\xi_C\,\restriction s)+2}.$
\end{clm}

\begin{proof} (Sketch)
\noindent  Note that  \,$\Phi_C(\emptyset)\leq 2$. \,Hence \,$\delta_C(\emptyset)\leq 2+2=2^2$.
\\ Now $\delta_C$ is a martingale, hence for $\,x\in\words,\ i\in\{0,1\}$,\, $\delta_C(xi)\leq 2\delta_C(x)$. By clause (b) in definition~\ref{def-xic}, there is an increase of $\delta_C(x)$ ($\leq$ than mutiplication by 2) only when case 1 occurs, and case 1 has occured $T(\xi_C\restriction s)$ times along $\xi\restriction s$.
 \end{proof}
 
 \begin{clm}\label{claim4}
 For any $e\in\n, \ e+T(\xi_C\restriction s)+2< h_e(s)$ \ a.e. (relatively to $s$)
 \end{clm}
 
 \begin{proof} (Sketch)\  \\
 Let $e\in\n$ be fixed. By Claim~\ref{claim2}, there is $s_0$ such that \,$T(\xi_C\restriction s_0)>e$. By Claim~\ref{claim1}, there is $s_1\geq s_0$ such that \,$\xi_C\restriction s_1$ \,is in case 1.
\\One then checks by induction on \,$s\geq s_1+1$ that \,$e+T(\xi_C\restriction s)+2<h_e(s)$.
\begin{itemize}[label=$-$]
\item
 If $\xi_C\restriction s$ is in case 1, then there is $m_e\leq s$ such that \ $e+T(\xi_C\restriction s)+3<h_e(m_e)\leq h_e(s)$.\\
 Hence \ $e+T(\xi_C\restriction {s+1})+2\,=\,e+T(\xi_C\restriction s)+3<h_e(s)\leq h_e(s+1)$.
\item If $\xi_C\restriction s$ is in case 2, this follows directly from the induction hypothesis.\qedhere
\end{itemize}
 \end{proof} 
 
\noindent One derives from the previous claims:
 
 \begin{clm}\label{claim5}
 For any $e\in\n, \ \  d_{C,e}(\xi_C\restriction s)<2^{h_e(s)}$ \ a.e.
 \end{clm}
 
 \begin{proof}\ \ For any $e,s\in\n$, \ \ \ $d_{C,e}(\xi_C\restriction s)\leq 2^e\Phi_C(\xi_C\restriction s)\leq 2^e\delta_C(\xi_C\restriction s)\leq 2^{e}2^{T(\xi_C\restriction s)+2}.$
\\ Since \ $e+T(\xi_C\restriction s)+2< h_e(s)$ \ a.e, we deduce \ $d_{C,e}(\xi_C\restriction s)<2^{h_e(s)}$ \ a.e.
 \end{proof}
 
\noindent One can now conclude: for our classes \,\C, any true \,\C\, order is a true \,\primrec\, order, hence by Claim~\ref{bound-inv}(a) and fact~\ref{enum-mart-order}, given any couple $(d,h)$ such that $d$ is a martingale in \,\C,\ with $d(\emptyset)\leq 1$ \,and $h$ is a true \,\C\ order, there is $e\in\n$ such that \,$d=d_{C,e}$\, and \,$h_e\leq h$ . Hence $\xi_C$ is \,\C\,-S random. By claim~\ref{claim2}, $\xi_C$ is not \,\pol\, random. Finally  Fact~\ref{complex-xic} gives the complexity of $\xi_C$.
 \end{proof}
 
 \subsection{Subcomputable weak randomness.}
 
 We compared the notion of \,\C\,-S randomness with the stronger notion of \,\C\ randomness. In this subsection, we shall study the relation  of \,\C\,-S randomness with the weaker notion of ``Kurz \,\C\ randomness". There are two candidates for the notion. Wang~\cite[Definition 5]{wang2000} proposed a notion in terms of martingales and orders:
 
 \begin{defi}[Wang]
 Let \,\C\,\ be a class of functions and let $\xi\in\cantor$.
\begin{itemize}[label=$-$]
\item If $d,\, h$ are respectively a martingale  and an order, then $\xi$ fails the Kurz test $\,(d,h)$\,\linebreak if \ $d(\xi\restriction i)\geq h(i)$\ \ a.e.
\item $\xi$ \,is \,(\C,\C)-W random if $\xi$ passes all Kurz tests for $d,h$ both in\, \C.
\end{itemize}
 \end{defi}
 
 Restricting to the classes \,\primrec\ \,or\, \pspace, Buss, Cenzer and Remmel~\cite{buss} proposed a different notion (called BP-randomness) and gave three different characterizations in terms of ML-tests, Kolmogorov complexity and martingale property. We give here
 the martingale and the ML test characterization in the primitive recursive context:
 
 \begin{thm}[Buss,\,Cenzer,\,Remmel]\label{BCR}
 Let $\xi\in\cantor$.\\
  $\xi$ is BP-random \ \ $\begin{array}[t]{cl}  
\text{iff} & \ \ \left\{ \begin{array}{l}
 \text{for no primitive recursive sequence } (U_n)_{n\in\n} \text{ of clopen sets}\\
\text{ such that }
\mu(U_n)\leq 2^{-n},\ \xi\in\bigcap_{n\in\n}U_n.\end{array}\right.\\
\text{iff}& \ \ \left\{ \begin{array}{l}
 \text{for no primitive recursive martingale } d
 \text{ and for no primitive }\\\text{recursive } \text{function }   f\in\n^\n,\ 
 d(\xi\restriction f(n))\geq 2^n \text{ \ a.e.}
 \end{array}\right.
 \end{array}$
 \end{thm}
 
 \begin{rem}\label{incr}
 In the theorem, one can replace ``no primitive recursive function $f$" by ``no primitive recursive strictly increasing function $f$": in the proof of~\cite[Thm 2.8 (2)$\Rightarrow$(1)]{buss}, the function $f$ can clearly be chosen strictly increasing.
 \end{rem}\vspace{-1mm}
 
Our notion of \,\C\,-S randomness cannot be compared with Wang's notion of (\C,\C)-W randomness since by~\cite[Cor 17]{wang2000}, for computable sequences, (\pol,\pol)-W randomness and p-randomness coincide.
~\cite{buss} notion is the right weakening of our notion of \,\C\,-S randomness:

\begin{fact}\label{weak-true}
Let $\xi\in\cantor$.\\
\espa $\xi$ is BP-random \ \
$ \Leftrightarrow \ \ \ \left\{ \begin{array}{l}
 \text{for no  martingale } d \text{ in \,\primrec\ }
 \text{ and for no true }\\\text{ \,\primrec\, order } h, \ 
 d(\xi\restriction n)\geq 2^{h(n)} \text{ \ a.e.}
 \end{array}\right.$
 \end{fact}
 
 \begin{proof}\ \\ $\pmb{\Rightarrow}\,:$ \ It is possible to derive this result from the ML characterization (see~\cite[7.2.13]{dohi} for a similar situation), but it can also be deduced from the above martingale characterization by the method of saving accounts (~\cite[6.3.8]{dohi}).
\\ Let $\xi\in\cantor$ and let $\ d,\ f$ be primitive recursive such that $f$ is strictly  increasing (Remark~\ref{incr}) and for  $n\geq n_0,\ d(\xi\restriction f(n))\geq 2^n$.
\\By considering the function $n\mapsto f(2n+1)$, we can assume \ for $n\geq n_0,\ d(\xi\restriction f(n))\geq 2^{2n+1}$.
\\ let us define inductively the primitive recursive $\q$-valued martingale $\delta$:
\begin{itemize}[label=$-$]
\item
  For $|x|\leq f(n_0)$, let \ $\delta(x)=d(x)$.
\item We assume now that $|x|\geq f(n_0)$ and for any $y\curl x,\ \delta(y)$ is defined.\\ Let \ $ n_0\leq n\leq |x|$ \  be such that \ $f(n)\leq |x|< f(n+1)$. For $i\in\{0,1\},$ 
 one sets
\centerline{ $\delta(xi)=\frac{\delta(x\,\restriction f(n))}{2}+(\delta(x)-\frac{\delta(x\,\restriction f(n))}{2})\frac{d(xi)}{d(x)}$.}
\end{itemize}
This defines a primitive recursive martingale and one checks by induction on $n\geq n_0$, that for any $x\in\words$, \ if $|x|=f(n)$,\, then \ $\frac{\delta(x)}{d(x)}\geq 2^{-(n-n_0)}$.
\\By definition of $\delta$, for any $x$ such that \ $f(n)\leq |x|\leq f(n+1)$, \ $\delta(x)\geq (1/2)\delta(x\,\restriction f(n))$.
\\ Hence one deduces for any $m,n$ such that \ $f(n)\leq m <f(n+1)$,

\centerline{$\delta(\xi\,\restriction m)\geq (1/2)\delta(\xi\,\restriction f(n))\geq 2^{-(n-n_0+1)}d(\xi\,\restriction f(n))\geq 2^{-(n-n_0+1)}2^{2n+1}\geq 2^n$.}
Let us define \,$h=\text{Inv}_f\dotminus 1$. \,Then $h$ is a true \,\primrec\, order because $f$ is strictly increasing.
If \,$m>f(n_0)$\, and \,$f(n)\leq m<f(n+1)$, \,then\, $0<\text{Inv}_f(m)\leq n+1$\, and hence \,$h(m)\leq n$.
 Therefore \,$\delta(\xi\,\restriction m)\geq 2^n\geq 2^{h(m)}$,\, for any \,$m>f(n_0)$.
(Rigorously $\delta$ is  $\q$-valued and not  $\q_2$-valued, but we can approximate it in $\q_2$ and apply Lemma~\ref{q2appro}).\\

\noindent $\pmb{\Leftarrow}$ :
 Let us assume \,$d(\xi\restriction n)\geq 2^{h(n)}$ \,a.e. for $h$ a true \,\primrec\, order. Then \,$\text{Inv}_h$ \,is primitive recursive.
For $k\in\n$, if $\,n_k=\text{Inv}_h(k)$, \,then \,$h(n_k)\geq k$. \,Hence for any $k\in\n,$
\espa\espa $ d(\xi\restriction \text
{Inv}_h(k))=d(\xi\restriction n_k)\geq 2^{h(n_k)}\geq 2^k.$
\end{proof}

Hence by the previous fact, \,\primrec\,-S randomness implies the BP-randomness of~\cite{buss}.

Schnorr~\cite{schno} showed that any p-random sequence $\xi$ satisfies the law of Large Numbers (that is if $s_n(\xi)=\sum_{k\leq n}\xi(k)$, then $\lim\limits_n s_n(\xi)/n =1/2$).  As noted in~\cite[Thm 5.1.8]{wang-phd}, his argument applies to \,(\pol,\pol)-S random sequences. This is also the case for \,\pol -S random sequences:

\begin{thm}[Schnorr]\label{large-numbers}\   \\
Every \,\pol\,-S random sequence satisfies the law of Large Numbers.
\end{thm}

\begin{proof}\ \ We refer to the exposition of Schnorr's theorem in~\cite[Thm 5.2.12]{wang-phd} and mention here only the small modification at the end.
Let us suppose $\xi\in\cantor$ does not satisfy the law of Large numbers. One can assume w.l.o.g.  \,$\limsup_n s_n(\xi)/n >1/2$. 
 \\Let $a=(\limsup_n s_n(\xi)/n )-(1/2)\,>0$ \ and \ 
let \,$q$\, in \,$\q_2\cap (0,1)$ \,be so that

\centerline{$\frac{1}{2}(\log(1+q)+\log(1-q))+a(\log(1+q)-\log(1-q))=c>0$.}
\noindent Defining the  $\q_2$-valued martingale $F$ in \,\pol\, as in~\cite{wang-phd} by:

\centerline{$\begin{array}{lcl}
F(\emptyset) &=& 1\\
F(x1)&=&(1+q)F(x)\\
F(x0)&=&(1-q)F(x),
\end{array}$}
\noindent one checks that \ $\limsup_n \log(F(\xi\restriction_n))/n\,=\,c$. \ 
Hence \ $\log(F(\xi\restriction_n))/n\,\geq c/2$ \ \  \emph{i.o.}
 Taking \,$k_0\in\n$ \,such that \,$1/k_0\leq c/2$\, and setting \ $h(n)=\lfloor n/k_0\rfloor$, for $n\in\n$, one obtains that \ $F(\xi\restriction n)\geq 2^{h(n)}$ \ \emph{i.o.} \ 
Since $\,h\,$ is a true \,\pol-order, \,$\xi$\, cannot be \,\pol -S random.
\end{proof}

On the opposite, it is known~\cite{kurz} that weak randomness does not imply satisfaction of the law of Large Numbers. In the context of primitive recursiveness, Buss, Cenzer and Remmel obtained the following:
\begin{thm}[{{~\cite[Thm 2.16]{buss}}}]\label{nolaw}
There exists a computable BP-random sequence which does not satisfy the law of Large Numbers.
\end{thm}

Hence \,\primrec -S randomness is strictly stronger than BP-randomness (even for computable sequences).

\subsection{A summary}

In order to summarize all results (some already known, some obtained here) in two tableaux, we agree on the following definitions (only the third one is new):

\begin{defi}\label{defi-agree}
For a class \,\C\, and \,$\xi\in\cantor$ (the martingales are \,$\q_2$-valued).

$\begin{array}{lcl}
\text{\point \ }\xi\text{ is \,\C\, random }\ & \text{iff}&\ \ \  \text{for no martingale $d$ in \,\C, \ $\limsup_{n\in\n} d(\xi\restriction n)=+\infty$.}\\
\text{\point \ $\xi$ is \,\C\,-S random }\ & \text{iff}& \ \begin{array}[t]{l}\text{for no martingale $d$ in \,\C\, and no true \,\C\, order $h$,}\\ d(\xi\restriction n)\geq 2^{h(n)}\text{ \ i.o}
\end{array}\\
\text{\point \ $\xi$ is \,\C\,-W random }\ & \text{iff}& \ \begin{array}[t]{l}\text{for no martingale $d$ in \,\C\, and no true \,\C\, order $h$,}\\ d(\xi\restriction n)\geq 2^{h(n)}\text{ \ a.e.}
\end{array}
\end{array}$
\end{defi}

\noindent If \,\C\, is the class of computable functions, then this corresponds to the classical notions of computable randomness, Schnorr randomness and weak randomness.

\begin{rem}\label{obvious}
If an infinite binary sequence \,$\xi$\, is in the class \,C\,, then $\xi$ is not \,\C\,-W random.
\end{rem}

To see this is true, we cannot simply say \ $\xi\in\bigcap_{n\in\n}[\xi\restriction n]$. The equivalence  between the ML definition and the martingale definition has only been  shown for the class \,\primrec\ (Theorem~\ref{BCR}) and may be problematic for low time-complexity classes.\\
To justify the remark, let us note that if \,$\xi$\, is in the class \,C\,, then one can consider the martingale $d$ defined as \ $d(x)=\sum_{i\in\n}2^i\mu([\xi\restriction 2i]|x)$.  $d$ is a $\q_2$-valued martingale in \,\C\, and for any 
$j\in\n,\ \ d(\xi\restriction j)\geq 2^{\lfloor j/2\rfloor}$.
The function \ $j\mapsto\lfloor j/2\rfloor$ \ is a true \,\pol\, order, hence $\xi$ is not \,\C\,-W random.

In the following tableau, no implication can be reversed:

\centerline{$\begin{array}{|ccccc|}\hline
&&&&\vspace{-2mm}\\
\!\text{Computable randomness}\!\!&\!\!\overset{\esp}{\Rightarrow}&\text{Schnorr randomness}&\Rightarrow&\text{weak randomness}
\\
&&&&\vspace{-3mm}\\
\text{\small{(1)}}\,\Downarrow\esp &&\text{\small{(1)}}\,\Downarrow\esp&&\text{\small{(1)}}\,\Downarrow\esp
\\
&&&&\vspace{-3mm}\\
\!\text{\primrec\ randomness}\!\!\!\!&\!\!\!\!\underset{\text{\small{(2)}}}{\Rightarrow}\!\!\!&\!\!\!\text{\primrec\,-S randomness}\!\!\!\!&\!\underset{\text{\small{(3)}}}{\Rightarrow}\!\!\!&\!\!\!\text{\primrec\,-W randomness\!}\\
\hline
\end{array}$}

The impossibility of reversing the implications in the first line is classical (Schnorr, Wang).
Concerning the other implications:
\begin{enumerate}
\item let   \ \C\ $\strin$\ \C '\ be two classes in our collection of time-complexity classes, or let \,C\, be in our collection and let \,\C ' be the class of recursive functions. It is known that there is a sequence $\xi$ in \,\C '\ which is \,\C\ random (One can use the martingale $\Psi_C(x)=\sum_e2^{-e}d_C(e,x)$, with $d_C$ given in Definition~\ref{easy}). By Remark~\ref{obvious}, $\xi$ is not \,\C '-W\ random,
\item is by Proposition~\ref{csrandom-not-crandom}(i),
\item is by Proposition~\ref{large-numbers} and Theorem~\ref{nolaw}.
\end{enumerate}

In the next tableau, the non-reversibility of the implication holds also when the notion is restricted to the class of computable sequences. (1),\,(2) and (3) refer to the above justifications (we abbreviate \toexp\ to \texp\ to be able to insert the tableau).\\


\centerline{$\begin{array}{|ccccc|}\hline
&&&&\vspace{-2mm}\\
\!\!\text{\primrec\ randomness}\!\!\!\!&\!\!\!\!\underset{\text{\small{(2)}}}{\Rightarrow}\!\!\!\!&\!\!\!\!\text{\primrec\,-S\ randomness}\!\!\!\!&\!\!\!\!\underset{\text{\small{(3)}}}{\Rightarrow}\!\!\!\!&\!\!\!\!\text{\primrec\,-W randomness}\!\vspace{-1mm}\\
\text{\small{(1)}}\,\Downarrow\esp &&\text{\small{(1)}}\,\Downarrow\esp&&\text{\small{(1)}}\,\Downarrow\esp\\
&&&&\vspace{-3mm}\\
\text{\texp\ randomness}\!\!\!&\!\!\!\underset{\text{\small{(2)}}}{\Rightarrow}\!\!\!&\!\!\!\text{\texp\,-S randomness}\!\!\!&\!\!\!\underset{\text{\small{(3)}}}{\Rightarrow}\!\!\!&\!\!\!\text{\texp\,-W randomness.}\vspace{-1mm}\\
\text{\small{(1)}}\,\Downarrow\esp &&\text{\small{(1)}}\,\Downarrow\esp&&\text{\small{(1)}}\,\Downarrow\esp\\
&&&&\vspace{-3mm}\\
\text{\expo\ randomness}&\underset{\text{\small{(2)}}}{\Rightarrow}&\text{\expo\,-S\ randomness}&\underset{\text{\small{(3)}}}{\Rightarrow}&\text{\expo\,-W randomness}\vspace{-1mm}\\
\text{\small{(1)}}\,\Downarrow\esp &&\text{\small{(1)}}\,\Downarrow\esp&&\text{\small{(1)}}\,\Downarrow\esp\\
&&&&\vspace{-3mm}\\
\text{\pol\ randomness}&\underset{\text{\small{(2)}}}{\Rightarrow}&\text{\pol\,-S\ randomness}&\underset{\text{\small{(3)}}}{\Rightarrow}&\text{\pol\,-W randomness}\\
\hline
\end{array}$}


\nocite{ch,ko}

\end{document}